\documentclass[12pt]{article}

\usepackage{a4wide}
\usepackage{amssymb}
\usepackage{amsmath,amsthm}
\usepackage[latin1]{inputenc}
\usepackage{graphicx}
\usepackage{hyperref}
\usepackage{enumerate}
\usepackage{tikz}
\usepackage{tkz-graph}
\usepackage{mathrsfs}
\usepackage{verbatim}
\usepackage{algorithm}
\usepackage{algorithmic}
\usepackage{mathtools}
\usepackage[normalem]{ulem}
\usepackage{subfigure}
\usepackage{xfrac}

\usepackage{tikz-cd}
\usetikzlibrary{decorations.markings}

\newcommand\mlnode[1]{\fbox{\begin{tabular}{@{}c@{}}#1\end{tabular}}}

\newtheorem{definition}{Definition}
\newtheorem{theorem}{Theorem}

\newtheorem{proposition}{Proposition}

\newtheorem{assumption}{Assumption} 

\hypersetup{colorlinks=true, linkcolor=blue, citecolor=blue, urlcolor=blue}

\newcommand{\iso}{\varphi}
\newcommand{\isiso}[1]{\simeq_{#1}}

\newcommand{\plussyb}[1]{#1^{+}}
\newcommand{\plussybg}{\plussyb{G}}
\newcommand{\pseudo}[1]{\iso#1}
\newcommand{\pseudog}{\pseudo{\plussybg}}
\newcommand{\transf}[1]{\transfName(#1)}
\newcommand{\transfg}{\transf{\pseudog}}
\newcommand{\transfName}{t}

\newcommand{\advsubgraph}[3]{#1_{#2, #3}} 

\newcommand{\advprobdist}{p}

\newcommand{\advguess}{\rho}
\newcommand{\advguessuni}{\Phi_{S,I}}

\newcommand{\equivclasspart}[1]{%
  #1/{\sim}%
}
\newcommand{\equivclassset}[1]{%
  [#1]_{\sim}%
}

\title{Preventing active re-identification attacks on social graphs 
via sybil subgraph obfuscation}

\author{Sjouke Mauw$^{1,2}$, Yunior Ram\'{i}rez-Cruz$^2$ 
and Rolando Trujillo-Rasua$^3$\\ 
{\small $^1$DCS, $^2$SnT, University of Luxembourg}\\ 
{\small 6, av. de la Fonte, L-4364 Esch-sur-Alzette, Luxembourg}\\ 
{\small $^3$School of Information Technology, Deakin University}\\ 
{\small 221 Burwood Hwy., Burwood VIC 3125, Australia}\\ 
{\small \texttt{\{sjouke.mauw, yunior.ramirez\}\@@uni.lu}, 
\texttt{rolando.trujillo\@@deakin.edu.au}}} 

\begin{document}
\maketitle

\begin{abstract}
Active re-identification attacks constitute a serious threat 
to privacy-preserving social graph publication, because of the ability 
of active adversaries to leverage fake accounts, a.k.a.\ \emph{sybil nodes}, 
to enforce structural patterns that can be used to re-identify their victims 
on anonymised graphs. Several formal privacy properties have been enunciated 
with the purpose of characterising the resistance of a graph against 
active attacks. However, anonymisation methods devised on the basis 
of these properties have so far been able to address only restricted 
special cases, where the adversaries are assumed to leverage 
a very small number of sybil nodes. 
In this paper we present a new probabilistic interpretation 
of active re-identification attacks on social graphs. 
Unlike the aforementioned privacy properties, which model the protection 
from active adversaries as the task of making victim nodes indistinguishable 
in terms of their fingerprints with respect to all potential attackers, 
our new formulation introduces a more complete view, where the attack 
is countered by jointly preventing the attacker from retrieving 
the set of sybil nodes, and from using these sybil nodes 
for re-identifying the victims. Under the new formulation, 
we show that the privacy property $k$-symmetry, 
introduced in the context of passive attacks, provides a sufficient condition 
for the protection against active re-identification attacks leveraging 
an arbitrary number of sybil nodes. Moreover, we show 
that the algorithm \mbox{\textsc{K-Match}}, originally devised 
for efficiently enforcing the related notion of $k$-automorphism, 
also guarantees $k$-symmetry. Empirical results on several collections 
of synthetic graphs corroborate that our approach allows, for the first time, 
to publish anonymised social graphs (with formal privacy guarantees) 
that effectively resist the strongest active re-identification attack reported 
in the literature, even when it leverages a large number of sybil nodes. 
 
\end{abstract}


\section{Introduction}\label{sec-intro}

The last decade has witnessed a formidable explosion in the use of social 
networking sites. Although the discipline of social network analysis 
has existed already for quite some time, today's scientists
potentially have access 
as never before to massive amounts of social network data. Social graphs 
are a particular example of this type of data, in which vertices typically 
represent users (e.g.\ Facebook or Twitter users, e-mail addresses) 
and edges represent relations between these users (e.g.\ becoming ``friends", 
following someone, exchanging e-mails). 
The analysis of social graphs can help scientists and other actors 
to discover important societal trends, study consumption habits, 
understand the spread of news or diseases, etc. For these goals 
to be achievable, it is necessary that the holders of this information, 
e.g.\ online social networks, messaging services, among others, 
release samples of their social graphs. However, ethical considerations, 
increased public awareness, and reinforced legislation\footnote{For example, 
the European GDPR, which can be consulted at \url{https://ec.europa.eu/commission/priorities/justice-and-fundamental-rights/data-protection/2018-reform-eu-data-protection-rules_en}.} 
place an increasingly strong emphasis on the need to protect individuals' 
privacy via anonymisation. 

Social graphs have proven themselves a challenging data type to anonymise. Even 
a simple undirected graph, with arbitrary node labels and no attributes 
on vertices or edges, is susceptible of leaking private information, 
due to the existence of unique structural patterns that characterise 
some individuals, e.g.\ the number of friends or the relations 
in the immediate vicinity \cite{NS2009}. Many privacy attacks that 
solely rely on the underlying graph topology of the social 
graph exist~\cite{Abawajy2016}, and they are still 
effective~\cite{MRT2018robsyb}, 
despite advances on social graph anonymisation.  
A particularly effective privacy attack is the so-called 
\emph{active attack}, which 
uses a strategy 
consisting in inserting fake accounts, commonly referred to as \emph{sybils}, 
into the real network. Once inserted, these fake users interact 
with legitimate users and among themselves, and create structures 
that allow the adversary to retrieve the sybil nodes from a sanitised social graph 
and use the connection patterns between sybils and legitimate nodes 
to re-identify the original users and infer sensitive information 
about them, such as the existence of relations. 

The publication of social graphs that effectively resist active 
attacks was initially addressed by Trujillo-Rasua and Yero~\cite{TY2016}. 
They introduced the notion of $(k,\ell)$-anonymity, the first privacy property 
to explicitly model the protection of published graphs against active adversaries. 
A graph satisfying $(k,\ell)$-anonymity ensures that an adversary leveraging 
up to $\ell$ sybil nodes and knowing the pairwise distances of all victims 
to all sybil nodes, is still unable to distinguish each victim 
from at least $k-1$ other vertices in the graph. This privacy property served 
as the basis for defining several anonymisation methods for a particular case, 
namely the one where either $k>1$ 
or $\ell>1$~\cite{MauwTrujilloXuan2016,MauwRamirezTrujillo2016}. 
In other words, non-trivial anonymity ($k>1$) was only guaranteed against 
an adversary leveraging exactly one sybil node. Later, the introduction 
of the notion of $(k,\ell)$-adjacency anonymity~\cite{MRT2018} allowed 
to arbitrarily increase the values of $k$ for which a formal privacy guarantee 
can be provided, but the proposed methods remained unable to address scenarios 
where the adversary can leverage more than two sybil nodes. 
In consequence, until now no anonymisation method with theoretically sound 
privacy guarantees against active attackers leveraging three or more sybil nodes 
has been made available to data publishers. 

To remedy the aforementioned situation, in this paper we take a different 
approach to the problem of anonymising a graph that may have been victim 
of an active attack. Our approach differs from the ones 
based on $(k,\ell)$-(adjacency) anonymity in the fact that it quantifies 
the combined probability of success of the attacker 
in re-identifying her sybil nodes \textbf{and} using them to re-identify 
the victims, whereas $(k,\ell)$-(adjacency) anonymity-based methods 
implicitly assume that the adversary cannot be prevented from retrieving 
the sybil nodes, and \mbox{(over-)compensate} for this by quantifying the probability 
that victims are re-identified in terms of \textbf{any} vertex subset 
satisfying some minimal constraints. Our new approach leverages 
a new probabilistic interpretation of the adversary's success probability 
in the two stages of re-identification. This new formulation allowed us to point out  
some mismatches between the goal of privacy-preserving publication, 
namely ensuring some upper bound on the probability 
of each victim being re-identified, and the actual guarantees provided 
by existing privacy properties. 
As an interesting observation, 
we noticed that there exist graphs that, despite failing to satisfy 
existing formal privacy properties, are in fact secure against active attackers. 
More importantly, our new formulation allowed us to prove 
that $k$-symmetry~\cite{Wu2010}, 
an existing privacy property originally introduced for the scenario 
of passive attacks, i.e.\ attacks that do not use sybil accounts, guarantees 
an $1/k$ upper bound on the re-identification probability of each victim, 
regardless of the number of sybil nodes used by the adversary. 
This finding is of paramount importance, 
as it enables the use of any known anonymisation method ensuring $k$-symmetry 
for preventing active re-identification attacks via sybil subgraph obfuscation. 
In this sense, we additionally show that 
the algorithm \mbox{\textsc{K-Match}}~\cite{ZCO2009}, originally devised 
for efficiently enforcing the notion of $k$-automorphism, also provides 
a sufficient condition for ensuring $k$-symmetry.  

\vspace{1mm}
\noindent
{\bf Summary of contributions:} 
\begin{itemize}
\item We show that no privacy property in the literature 
characterises all anonymous graphs with respect to active attacks. 
\item We introduce a general definition of resistance to active attacks 
that can be used to analyse the actual resistance of a graph. 
\item We use the introduced privacy model to prove that $k$-symmetry, the 
strongest notion of anonymity against passive attacks, also protects 
against active attacks. 
\item Of independent interest is our proof that $k$-automorphism does not protect 
against active attacks. This is a surprising result, considering that 
$k$-automor\-phism and $k$-symmetry have traditionally been deemed 
as conceptually equivalent.
\item We prove that the algorithm \textsc{K-Match}, devised to ensure 
a sufficient condition for $k$-automorphism, also guarantees $k$-symmetry. 
\item We provide empirical evidence on the effectiveness of \textsc{K-Match} as an 
anonymisation strategy against the strongest active attack reported 
in the literature, namely the robust active attack presented 
in~\cite{MRT2018robsyb}, even when it leverages a large number of sybil nodes.
\end{itemize}

\vspace{1mm}
\noindent
{\bf Structure of the paper.} We discuss related work 
in Section~\ref{sec-related-work}, and describe our new probabilistic 
interpretation of the adversarial model for active re-identification attacks 
in Section~\ref{sec-adversarial-model}. Then we discuss the applicability 
of $k$-symmetry for modelling protection against active attackers 
in Section~\ref{sec-k-symmetry}, and show in Section~\ref{sec-k-match} 
that the algorithm \textsc{K-Match} efficiently provides a sufficient condition 
for $k$-symmetry. Finally, we empirically demonstrate the effectiveness 
of \textsc{K-Match} against the robust active attack from~\cite{MRT2018robsyb}
in Section~\ref{sec-experiments} and give our conclusions 
in Section~\ref{sec-concl}.

\section{Related work} 
\label{sec-related-work}

In this paper we focus on a particular family of properties 
for privacy-preserving publication of social graphs: 
those based on the notion of $k$-anonymity~\cite{S2001,Sweeney2002}. 
These privacy properties depend on assumptions about the type of knowledge 
that a malicious agent, the \emph{adversary}, possesses. 
According to this criterion, adversaries can be divided into two types. 
On the one hand, \emph{passive} adversaries rely on information 
that can be collected from public sources, such as public profiles 
in online social networks, where a majority of users keep unmodified 
default privacy settings that pose no access restrictions on friend lists 
and other types of information. A passive adversary attempts
to re-identify users in a published social graph by matching this information 
to the released data. On the other hand, \emph{active} adversaries not only 
use publicly available information, but also attempt to interact 
with the real social network before the data is published, with the purpose 
of forcing the occurrence of unique structural patterns that can be retrieved 
after publication and used for learning sensitive information. 

\vspace{1mm}
\noindent 
{\bf $k$-anonymity models against passive attacks.} 
$k$-anonymity is based on a notion of indistinguishability  
between users in a dataset, which is used to create equivalence 
classes of users that are pair-wise indistinguishable to the eyes of an attacker.  
Formally, given a symmetric, reflexive and transitive indistinguishability 
relation~$\sim$ on the users of a graph $G$, $G$ satisfies $k$-anonymity 
with respect to $\sim$ if and only if the equivalence class 
with respect to $\sim$ of each user in $G$ has cardinality at~least~$k$. 

Several graph-oriented notions of indistinguishably appear in the literature. 
For example, Liu and Terzi~\cite{LT2008} consider two users indistinguishable 
if they have the same degree. Their model is known as \emph{$k$-degree anonymity} 
and gives protection against attackers capable of accurately estimating the number 
of connections of a user. The notion of $k$-degree anonymity 
has been widely studied, and numerous anonymisation methods 
based on it have been proposed, 
e.g.~\cite{Lu2012,UMGA,Chester2013,Wang2014,Ma2015,Salas2015,Rousseau2017,Casas-Roma2017}. 
Zhou and Pei~\cite{ZP2008} assume a stronger attacker 
able to determine not only the connections of a user $u$, 
but also whether $u$'s friends (i.e.\ those users that $u$ is connected to)
are connected. 
This means that the adversary is assumed to know the induced subgraphs created 
by the users and their neighbours. 
It is simple to see that Zhou and Pei's model, known as 
\emph{$k$-neighbourhood anonymity}, is stronger than $k$-degree anonymity. 

The notion of \emph{$k$-automorphism}~\cite{ZCO2009} was introduced 
with the goal of modelling the knowledge of any passive adversary. 
Two users $u$ and $v$ in a graph $G$ are said to be automorphically equivalent, 
or indistinguishable, if $\varphi(u) = v$ for some automorphism $\varphi$ in $G$. 
The notion of $k$-automorphism ensures that every vertex in the graph 
is automorphically equivalent to $k-1$ other vertices. 
Although $k$-automorphism itself does not in general imply 
all other privacy properties (as we will show in Appendix~\ref{app-automorphism}), 
the method proposed in~\cite{ZCO2009} for enforcing the (stronger) 
\emph{$k$ different matches principle} does achieve this goal. 
Similar formulations of indistinguishability in terms of graph automorphisms 
were presented independently in the work on $k$-symmetry~\cite{Wu2010} 
and $k$-isomorphism \cite{Cheng2010}. While $k$-symmetry 
and $k$-automorphism have traditionally been viewed as equivalent, 
$k$-symmetry is actually stronger, and it does imply all other privacy properties 
for passive attacks. In this paper, we additionally show that, in the context 
of active attacks, $k$-symmetry always guarantees a $1/k$ 
upper bound on the re-identification probability for each vertex, 
which $k$-automorphism does not. 

A natural trade-off between the strength 
of the privacy notions and the amount of structural disruption caused 
by the anonymisation methods based on them has been empirically demonstrated 
in~\cite{ZCO2009}. 
The three privacy models described above form a hierarchy, which is displayed 
in the left branch of Figure~\ref{fig-hierarchy}. 
Privacy models tailored to active attacks also form a hierarchy, 
displayed in the right branch of Figure~\ref{fig-hierarchy}, 
which we describe next. Interrogation marks in~Figure~\ref{fig-hierarchy} 
indicate that connections between properties tailored for passive attacks 
and those tailored for active attacks have not been established yet, 
neither directly nor via some additional property. 

\begin{figure}
	\centering
	\begin{tikzcd}[arrows=Rightarrow, column sep=1cm, row sep=1.0cm]
		&
		\mlnode{?}
		&
		\\
		\mlnode{$k$-symmetry} \arrow[d] \arrow[ur, Leftrightarrow, red, "?"] \arrow[rr, Leftrightarrow, red, "?"]  
		&      
		& 
		{\mlnode{$(k, \ell)$-anonymity}} \arrow[d] \arrow[ul, Leftrightarrow, red, swap, "?", shift right=0ex]
		\\
		\mlnode{$k$-neighbourhood \\ anonymity} \arrow[d]          
		& 
		& 
		{\mlnode{$(k, \ell)$-adjacency \\ anonymity}}                  
		\\
		\mlnode{$k$-degree anonymity}                   
		&      
		&           
	\end{tikzcd}
\caption{A hierarchy of privacy properties. An arrow has the standard logical 
interpretation, i.e.\ $P \implies P'$ means that a graph satisfying $P$ 
also satisfies $P'$. \textbf{Left side:} models for passive attacks. 
\textbf{Right side:} models for active attacks. Interrogation marks indicate 
connections that have not been established yet. 
\label{fig-hierarchy}}
\end{figure}
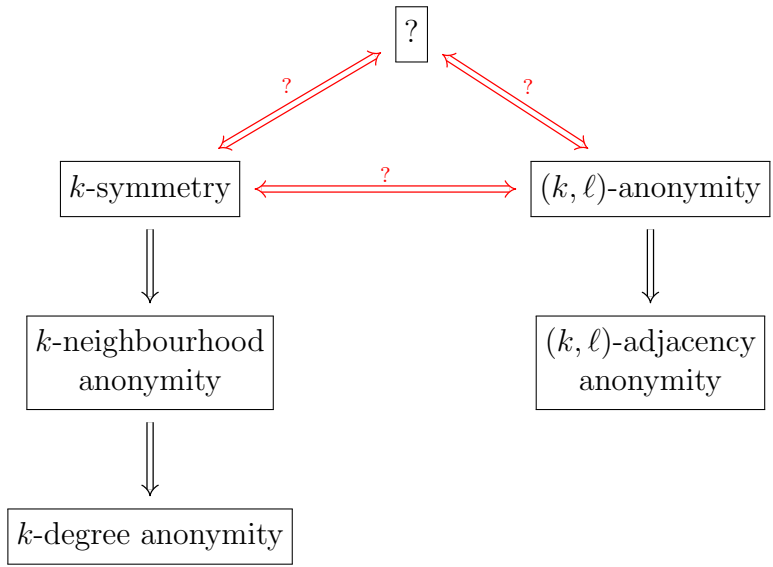

\vspace{1mm}
\noindent 
{\bf $k$-anonymity models against active attacks.} 
Backstrom et al.\ were the first 
to show the impact of active attacks 
in social networks back in 2007~\cite{BDK2007}. Their attack has been optimised 
a number of times, see~\cite{Peng2012,PLZW2014,MRT2018robsyb}, and two privacy 
models particularly tailored to measure the resistance of social graphs 
to this type of attack have been recently proposed~\cite{TY2016,MRT2018}. 
The first of those models is \emph{$(k, \ell)$-anonymity}, introduced in 2016 
by Trujillo-Rasua and Yero~\cite{TY2016}. They consider adversaries capable of 
re-identifying their own sets of sybil nodes in the anonymised graph. 
Adversaries are also assumed to know or able to estimate the distances of the
victims to 
the set of sybil nodes. This last assumption was weakened later 
in~\cite{MRT2018} by restricting the adversary's knowledge to distances between 
victims and sybil nodes of length one. That is, the adversary only knows 
whether the victim is connected to a sybil node. That restriction led to a 
weaker version of $(k, \ell)$-anonymity called \emph{$(k, \ell)$-adjacency 
anonymity}, as displayed in Figure~\ref{fig-hierarchy}. 

There exist three anonymisation 
algorithms~\cite{MauwTrujilloXuan2016,MauwRamirezTrujillo2016,MRT2018} that aim 
to create graphs satisfying $(k, \ell)$-(adjacency) anonymity. 
Their approach consists in determining a candidate set of sybil vertices 
in the original graph that breaks the desired anonymity property, 
and forcing via graph transformation that every vertex has a common pattern 
of connections with the sybil vertices shared by at least $k-1$ other vertices. 
A common shortcoming of these methods is that they only provide formal 
guarantees against attackers leveraging a very small number of sybil nodes 
(no more than two). This limitation 
seems to be an inherent shortcoming of the entire family of properties 
of which $(k,\ell)$-anonymity and $(k,\ell)$-adjacency anonymity 
are members. Indeed, for large values of $\ell$, which are required 
in order to account for reasonably capable adversaries, anonymisation methods 
based on this type of property face the problem that any change introduced 
in the original graph to render one vertex indistinguishable from others, 
in terms of its distances to a vertex subset, is likely to render this vertex unique 
in terms of its distances to other vertex subsets. 

An approach that has been used against both types of attack, passive and 
active, consists in randomly perturbing the graph. While intuition suggests 
that the task of re-identification becomes harder for the adversary 
as the amount of random noise added to a graph grows, no theoretically sound 
privacy guarantees have been provided for this scenario. 
Moreover, it has been shown in~\cite{MRT2018robsyb} that active attacks 
can be made robust against reasonably large amounts of random perturbation. 

\vspace{1mm}
\noindent 
{\bf Other privacy models.} For the sake of completeness, we finish 
this brief literature survey by mentioning privacy models that aim to quantify 
the probability that the adversary learns any sort of sensitive information. 
A popular example is differential privacy (DP)~\cite{DworkR14}, 
a semantic privacy notion 
which, instead of anonymising the dataset, focuses on the methods accessing 
the sensitive data, and provides a quantifiable privacy guarantee 
against an adversary who knows all but one entry in the dataset. 
In~the context of graph data, the notion of two datasets 
differing by exactly one entry can have multiple interpretations, 
the two most common being \emph{edge-differential privacy} 
and \emph{vertex-differential privacy}. 
While a number of queries, e.g.\ degree sequences~\cite{Karwa2012,HLMJ09} 
and subgraph counts~\cite{Zhang2015,KRSY14}, 
have been addressed under (edge-)differential privacy, the use of this notion 
for numerous very basic queries, e.g.\ graph diameter, remains a challenge. 
Recently, differentially private methods leveraging the randomized response strategy 
for publishing a graph's adjacency matrix were proposed in~\cite{ST2020}. 
While these methods do not necessarily view vertex ids as sensitive, data holders 
whose goal in preventing re-identification attacks 
is to prevent the adversary from learning the existence of relations may view 
this approach as an alternative to $k$-anonymity-based methods. 
Another DP-based alternative to $k$-anonymity-based methods consists 
in learning the parameters of a graph generative model under differential privacy 
and then using this model to publish synthetic graphs that resemble 
the original one in some structural 
properties~\cite{MW09,Sala2011,WW13,XCT14,Jorgensen2016,CMR2020}. 

Other privacy models  rely on the notion of \emph{adversary's prior belief}, 
defined as a probability distribution on sensitive values. 
For example, $t$-closeness~\cite{NTS2007} measures attribute 
protection in terms of the distance between the distribution 
of sensitive values in the anonymised dataset with respect to the distribution 
of sensitive attribute values in the original table. Such definition of prior 
belief is different to other works, 
such as $(\rho_1, \rho_2)$-privacy~\cite{EGS2003} 
and $\epsilon$-privacy~\cite{MGG2009}, where prior belief 
represents the adversary's knowledge in the absence of knowledge about the 
dataset. In either case, estimating the prior belief of the adversary 
is challenging, as discussed in~\cite{DworkR14}. 

\vspace{1mm}
\noindent 
{\bf Concluding remarks.} As illustrated in 
Figure~\ref{fig-hierarchy}, the development of $k$-anonymity models against 
passive and active attacks have been traditionally split and 
had no apparent intersection. This article provides, to the best of our 
knowledge, the first connection between the two developments. This is achieved 
by introducing a probabilistic model for active attacks that characterises all 
graphs that resists active attacks, of which $k$-symmetry and $(k, 
\ell)$-anonymity are proven to be sufficient, yet not necessary, conditions.

\section{Probabilistic adversarial model}
\label{sec-adversarial-model}

Our adversarial model is a generalisation of the model introduced in 
\cite{MRT2018robsyb}, which captures 
the capabilities of an active attacker and allows one to analyse the
resistance of anonymisation methods to active attacks. Such analysis 
is expressed as a three-step game 
between the attacker and the defender. In the first step the 
attacker is allowed to interact with 
the network, insert sybil accounts, and establish links with other users 
(called the victims). The defender uses the second step to 
anonymise and perturb the network, which was previously manipulated by the 
attacker. Lastly, the attacker receives the anonymised 
network and makes a guess on the pseudonyms used to anonymise the victims. 
Each of these steps is formalised in what follows. 

\begin{figure}[!th]
	\newcommand{\scaleValue}{1}
	\newcommand{\spaceValue}{2cm}
	\centering
	\begin{tabular}{ccc}
		\subfigure[A fragment of the original graph.\label{fig-original}]
		{\includegraphics[scale=\scaleValue]{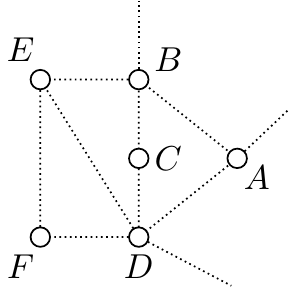}}&
		\quad\quad\quad\quad&
		\subfigure[Sybil nodes are added and victim fingerprints are created. 
		\label{fig-sybils}]
		{\includegraphics[scale=\scaleValue]{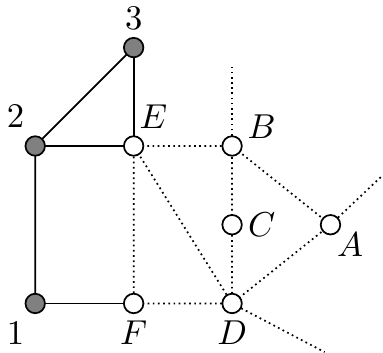}}\\
		\subfigure[Pseudonymised graph is released.]
		{\includegraphics[scale=\scaleValue]{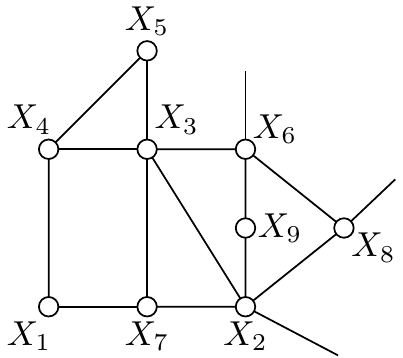} 
			\label{fig-anonymise}}&&
		\subfigure[Attacker subgraph is recovered, victims are re-identified, 
		and the existence of a relation is revealed. \label{fig-reidentification}]
		{\includegraphics[scale=\scaleValue]{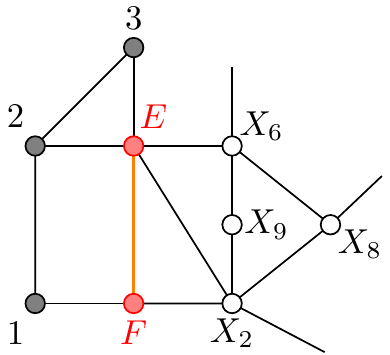}}
	\end{tabular}
	\caption{An active re-identification attack viewed as an attacker-defender game. 
	} 
	\label{fig-active-attack}
\end{figure}

\subsection{Attacker subgraph creation}

The attacker-defender game starts with a graph $G=(V,E)$ representing a 
snapshot of a social network, as in Figure~\ref{fig-original}. The attacker 
knows a subset of the users, called the victims and denoted $I$, 
but not the connections between them. The attacker is allowed to insert a set 
of sybil nodes $S$ into $G$ and establish connections with their victims. 

This step of the attack transforms the original graph $G=(V,E)$ into a 
graph $\plussybg =(V', E')$ satisfying the following two properties: i) $V' = 
V 
\cup 
S$ and ii) $E' \setminus E \subseteq (S\times 
S) \cup (S\times I) \cup (I \times S)$. The second condition says that 
relations 
established by 
the adversary are 
constrained to 
the set 
of sybil and victim nodes. 
We 
call the resulting graph  $\plussybg$ the 
\emph{sybil-extended} graph of $G$. 
An example of a sybil-extended graph is depicted 
in Figure~\ref{fig-sybils}.

The attacker does not know the entire graph 
$\plussybg$, unless the original graph was empty. The adversary knows, however, 
the subgraph formed by the set of sybil nodes $S$, their connections to the 
victims, and the victim set $I$. This notion of adversary knowledge is 
formalised next.

\begin{definition}[Adversary knowledge]\label{def-advknowledge}
Let $G = (V, E)$ be an original graph and $\plussybg =(V \cup S, E')$ the 
sybil-extended graph created by an adversary that targets a set of 
victims $I \subseteq V$. The \emph{adversary knowledge} 
is defined as the subgraph $\advsubgraph{G}{S}{I}$ of $G$ defined by 
\[
\advsubgraph{G}{S}{I} = 
(S \cup I, 
\{(u,v)\in E 
\mid \{u, v\} \subseteq S \cup I \wedge \{u, v\} \not\subseteq I\})
\]
\end{definition}

Note that connections between victims are not part of the adversary 
knowledge. 

\subsection{Pseudonymisation and perturbation}

When the defender decides to publish the graph $\plussybg$, she 
pseudonymises it by replacing the real user identities with pseudonyms. 
That is to say, the defender 
obtains $\plussybg$ and constructs an isomorphism 
$\iso$
from $\plussybg$ to $\pseudog$. 
An \emph{isomorphism} between two graphs $G=(V,E)$ and $G'=(V',E')$ is
a bijective function $\iso\colon V \to V'$, such that
$\forall v_1,v_2\in V \colon (v_1,v_2)\in E \iff (\iso(v_1),\iso(v_2))\in E'$.
Two graphs are \emph{isomorphic}, denoted by $G\isiso{\iso} G'$, or
briefly $G\isiso{} G'$, if there exists an isomorphism $\iso$ between them. 
Given a subset of vertices $S \subseteq V$, we will often use $\iso S$ to 
denote the set $\{\iso(v) | v \in S\}$. In Figure~\ref{fig-anonymise} 
we illustrate a pseudonymisation of the graph 
in Figure~\ref{fig-sybils}. 

We call $\pseudog$ the \emph{pseudonymised} graph. 
Pseudonymisation serves the purpose of removing 
personally identifiable information from the graph. Because pseudonymisation is 
insufficient to protect a graph against re-identification, 
the 
defender is also allowed to perturb the graph. This is captured by a 
non-deterministic procedure 
$\transfName$ 
that maps graphs to graphs. The procedure $t$ modifies $\pseudog$, 
resulting in the \emph{transformed} graph $\transfg$. We assume that $\transfg$ 
is ultimately made available to the public, hence it 
is known to the adversary. 

\subsection{Re-identification}

The last step of the attacker-defender game is where the attacker analyses the 
published graph $\transfg$  to re-identify her own sybil accounts and
the victims (see Figure~\ref{fig-reidentification}). 
This allows her to acquire new information, which was supposed 
to remain private, such as the fact that $E$ and $F$ are friends.

We define the output of the adversary re-identification attempt as a 
mapping $\advguess$ from the set of vertices $S \cup I$ to the set 
of vertices in $\transfg$. This represents the adversary's belief on the 
pseudonyms used to anonymise the attacker and victim vertices in $\transfg$. 
To 
account for uncertainty on the adversary's belief, 
we consider that the adversary assigns a probability value 
$\advprobdist(\advguess)$ to each mapping, denoting the probability that the 
adversary chooses $\advguess$ as the output of the re-identification attack. 
Let $\advguessuni$ be the universe of mappings from the set of vertices 
in $S \cup I$ to the set $\transfg$. 
The law of total probability allows us to quantify the adversary's 
probability of success in re-identifying one victim as follows. 

\begin{proposition}\label{prop-success-prob}
Let $G = (V, E)$ be an original graph, $\plussybg =(V \cup S, E')$ the 
sybil-extended graph created by an adversary that targets a set of 
victims $I \subseteq V$, and $\transfg$ the anonymised version of 
$\plussybg$ created by the defender. Then, the probability $A_{\transfg}^{S,I}(u)$ 
that the adversary successfully re-identifies a victim $u \in I$ in $\transfg$ is 
%

\begin{equation}\label{eq-success-prob-general}
A_{\transfg}^{S,I}(u) = \sum_{\advguess \in \advguessuni, \advguess(u) = 
\iso(u)} 
\advprobdist(\advguess).
\end{equation}
\end{proposition}

In our analyses, we restrict the function $\advprobdist$ 
to be a probability distribution on the domain $\advguessuni$, 
i.e. $\sum_{\advguess \in \advguessuni} \advprobdist(\advguess) = 1$. 
We also assume that $\advprobdist$ satisfies the 
standard \emph{random worlds assumption} enunciated in~\cite{CLR2007,MKMGH2007}, 
which expresses that, in the absence of any information in addition to $\transfg$, 
any two isomorphic subgraphs in  $\transfg$ are indistinguishable for the 
adversary. We enunciate the random worlds assumption next, 
adapting the terminology to the one used in this paper. 

\begin{assumption}[Random worlds assumption \cite{CLR2007,MKMGH2007}] 
\label{pro-random-worlds}
Let $G = (V, E)$ be an original graph, $\plussybg =(V \cup S, E')$ 
the sybil-extended graph created by an adversary that targets 
a set of victims $I \subseteq V$, and $G' = \transfg$ the anonymised version 
of $\plussybg$ created by the defender. 
Let $\advguess_1$ and $\advguess_2$ be 
two bijective functions from $S\cup I$ to 
the set of vertices $V_{G'}$ in $G'$. Let 
$\advsubgraph{G'}{\advguess_1 S}{\advguess_1 I}$ and 
$\advsubgraph{G'}{\advguess_2 S}{\advguess_2 I}$ be the two attacker subgraphs 
in $G'$ that correspond to the adversary's guesses $\advguess_1$ and 
$\advguess_2$, respectively. If  $\advsubgraph{G'}{\advguess_1 S}{\advguess_1 
I}$ and 
$\advsubgraph{G'}{\advguess_2 S}{\advguess_2 I}$ are isomorphic, then 
$\advprobdist(\advguess_1) 
= \advprobdist(\advguess_2)$. 
\end{assumption}

In the remainder of this article, we will analyse the effectiveness of various 
anonymisation procedures by calculating the success probability of the 
adversary based on Proposition~\ref{prop-success-prob}, and we will often 
resort to Assumption~\ref{pro-random-worlds} when reasoning about the adversary's 
belief $\advguess$.

\section{Applicability of current privacy properties against active attacks}
\label{sec-k-symmetry}

In this section we make, to the best of our knowledge, the first connection 
between passive and active attacks by formally proving that $k$-symmetry 
provides protection against active attacks. We also prove that $k$-symmetry is 
incomplete, just like $(k, \ell)$-anonymity, in the sense that none of them 
characterises all anonymous graphs with respect to active attacks. Last, but 
not least, we show that neither $k$-symmetry implies $(k, 
\ell)$-anonymity, nor the other way around. 

\subsection{$k$-symmetry: an effective privacy model against active attacks}

We use the introduced privacy model to prove that $k$-symmetry, the strongest 
notion of anonymity against passive attacks, also protects against active 
attacks.

\begin{definition}[$k$-symmetry \cite{Wu2010}]\ 
Let $\Gamma_G$ be the universe of automorphisms in $G$. Two vertices $u$ and 
$v$ in $G$ are said to be automorphically 
equivalent, denoted $u \cong v$, if there exists an automorphism 
$\gamma \in \Gamma_G$ such that $\gamma(u) = v$. 
Because the relation $\cong_{}$ is an equivalence relation in the set of 
vertices of $G$, let $[u]_{\cong}$ be the equivalence class of $u$. 
$G$ is said to satisfy $k$-symmetry if for every vertex $u$ it 
holds 
that $|[u]_{\cong}| \geq k$.
\end{definition}

\begin{theorem}\label{theo-k-symmetry}
	Let $G' = (V', E')$ be an original graph, $\plussybg =(V' \cup S, E')$ the 
	sybil-extended graph created by an adversary that targets a set of 
	victims $I \subseteq V'$, and $\transfg = (V, E)$ the anonymised version of 
	$\plussybg$ created by the defender. 
	If $\transfg$ satisfies $k$-symmetry, then for every vertex $u \in I$ the 
	probability of the adversary guessing the output of $\iso(u)$ is lower than 
	or 
	equal to $1/k$.
\end{theorem}
\begin{proof}
	Let $G$ be a shorthand 
	notation for $\transfg$. Let $\advguessuni$ be the universe of mappings 
	from 
	the set of vertices 
	in 
	$S \cup I$ to the set of vertices in $G$.
	We define a relation $\sim$ between adversary's guesses in 
	$\advguessuni$ by 
	
	\[\advguess \sim 
	\advguess' \iff \advsubgraph{G}{\advguess S}{\advguess I} \isiso{} 
	\advsubgraph{G}{\advguess' S}{\advguess' I}
	\]

	Because $\isiso{}$ is an equivalence relation, it 
	follows that $\sim$ is also an equivalence relation. We use 
	$\equivclasspart{\advguessuni}$ to denote the partition of $\advguessuni$ 
	into the
	set of equivalence classes with respect to $\sim$, and 
	$\equivclassset{\advguess}$ to denote the equivalence class containing 
	$\advguess$. Consider, given a victim $u$, a successful adversary guess 
	$\advguess_0 \in \advguessuni$, i.e.\ a mapping satisfying that 
	$\advguess_0(u) 
	= 
	\varphi(u)$. 
	Our first proof step is about showing that there exist $k-1$ other 
	mappings $\advguess_1, \ldots, \advguess_{k-1}$ in 
	$\equivclassset{\advguess}$ 
	satisfying that
	
	\begin{equation}\label{eq1-sym}
	\forall i, j \in \{0, \ldots, k-1\} \colon i \neq j \implies \advguess_i(u) 
	\neq 
	\advguess_j(u) \text{.}
	\end{equation}
	
	Let $\advguess_0(u) = v$. 
	Because $G$ satisfies $k$-symmetry, it follows that there exist $k-1$ 
	different vertices $\{v_1, \ldots, v_{k-1}\}$ that are automorphically 
	equivalent to $v$. That is to say, there exist $k-1$ automorphisms 
	$\gamma_1, 
	\ldots, \gamma_{k-1}$ in $\Gamma_G$ such that $\forall 
	i \in \{1, \ldots, k-1\} \colon \gamma_i(v) = v_i \neq v$. Now, 
	consider the mappings 
	$\advguess_i\colon S \cup I \rightarrow 
	S_i \cup I_i$ defined by $\advguess_i = \gamma_i \circ \rho_0$, for 
	every $i \in \{1, \ldots, k-1\}$. On the one hand, given that 
	$\gamma_1, \cdots, \gamma_{k-1}$ are automorphisms, it follows that 
	$\advsubgraph{G}{S_0}{I_0} \isiso{\gamma_i} 
	\advsubgraph{G}{S_i}{I_i}$, for every $i \in \{1, \ldots, k-1\}$, which 
	implies 
	that $\advguess_0 \sim \advguess_i$. On the other hand, 
	$\rho_i(u) = u_i \neq u_j = \rho_j(u)$ for every $i \neq j \in \{0, 
	\ldots, k-1\}$. This allows us to conclude that $\rho_0, \ldots, 
	\rho_{k-1}$ 
	are pairwise different and that $\{\rho_0, \ldots, \rho_{k-1}\} \subseteq 
	\equivclassset{\advguess}$.
	
	Our second proof step consists of showing that, given two mappings 
	$\advguess_0$ and $\advguess_0'$ in  $\advguessuni$ such that 
	$\advguess_0(u) = 
	\advguess_0'(u) = v$, and the mappings 
	$\{\advguess_1, \ldots, 
	\advguess_{k-1}\}$ and $\{\advguess_1', \ldots, 
	\advguess_{k-1}'\}$ constructed as previously, it holds that 
	
	\[
	\advguess_0 \neq \advguess_0' \implies \advguess_i \neq \advguess_j' \ 
	\forall 
	i, j \in \{1, 
	\ldots, k-1\} \text{.}
	\]

	Let $x \in S \cup I$ such that $\rho_0(x) 
	\neq \rho_0'(x)$. Take any two integers $i, j \in \{1, 
	\ldots, k-1 \}$. We analyse two cases.
	
	\noindent \emph{Case 1 $(i = j)$.}
	Let $\rho_0(x) = y$ and $\rho_0'(x) = y'$. By construction, $\rho_i(x) = 
	\gamma_i(\rho_0(x)) = \gamma_i(y)$ and $\rho_i'(x) = 
	\gamma_i(\rho_0'(x)) = \gamma_i(y')$. The fact that $\gamma_i$ is a 
	bijective 
	function and that $y \neq y'$ gives that $\gamma_i(y) \neq \gamma_i(y')$, 
	which 
	implies that $\rho_i \neq \rho_i'$.
	
	\noindent \emph{Case 2 $(i \neq j)$.}
	Observe that $\rho_i(u) = \gamma_i(\rho_0(u)) = \gamma_i(v) = v_i$ and 
	$\rho_j'(u) = 
	\gamma_j(\rho_0'(u)) = \gamma_j(v) = v_j$. Because $v_i \neq v_j$ it 
	follows 
	that $\rho_i(u) \neq \rho_j'(u)$, hence $\rho_i \neq \rho_j'$. 
	
	The last proof step consists of using the formula to calculate adversary 
	success to obtain a probability bound. The adversary's 
	probability of success in re-identifying a victim $u \in I$ is calculated 
	by,
	
	\[
	\sum_{\advguess \in \advguessuni, \advguess(u) = \iso(u)} 
	\advprobdist(\advguess).
	\]
	
	Let $\advguess_1^0, \ldots, \advguess_n^0$ be all functions in 
	$\advguessuni$ 
	satisfying $\advguess_1^0(u) = \advguess_2^0(u) = \cdots = \advguess_n^0(u) 
	= 
	\iso(u)$. It follows that the probability of success of the adversary is 
	equal 
	to $\advprobdist(\advguess_1^0) + \cdots + \advprobdist(\advguess_n^0)$. 
	Now, for each $\advguess_i$, consider the mappings $\advguess_i^1, \ldots, 
	\advguess_i^{k-1}$ defined by $\advguess_i^j = \varphi_j \circ 
	\advguess_i^0$, 
	for 
	every $j \in \{1, \ldots, k-1\}$. Previously we proved the following two 
	intermediate results. 
	
	\begin{enumerate}
		\item For every $i \in \{1, \ldots, n\}$, the set $\{\advguess_i^0, 
		\advguess_i^1, \ldots, \advguess_i^{k-1}\} \subseteq \advguessuni$ has 
		cardinality $k$ 
		and its 
		elements satisfy $\advguess_i^0 \sim
		\advguess_i^1 \sim \ldots \sim \advguess_i^{k-1}$. 
		\item $\forall i, j \in \{1, \ldots, n\} \implies \{\advguess_i^0, 
		\advguess_i^1, \ldots, \advguess_i^{k-1}\}\cap \{\advguess_j^0, 
		\advguess_j^1, \ldots, \advguess_j^{k-1}\} = \emptyset$.
	\end{enumerate}
	
	The 
	second result and the fact that $\advprobdist$ is a probability 
	distributions 
	gives, 
	
	\[
	\sum_{i \in \{1, \ldots, n\}} \sum_{j \in \{0, \ldots, k-1\}} 
	\advprobdist(\advguess_i^j) \leq 1 \text{.}
	\]
	
	We use the first result and the random 
	worlds assumption (Proposition~\ref{pro-random-worlds}) to conclude that 
	$\advprobdist(\advguess_i^0) = \advprobdist(\advguess_i^1) = \cdots = 
	\advprobdist(\advguess_i^{k-1})$, for every $i \in \{1, \ldots, n\}$, which 
	gives, 
	
	\[
	\sum_{i \in \{1, \ldots, n\}} \sum_{j \in \{0, \ldots, k-1\}} 
	\advprobdist(\advguess_i^j) = \sum_{i \in \{1, \ldots, n\}} k 
	\advprobdist(\advguess_i^0)
	\leq 1 \text{.}
	\]
	
	The last inequality states that $\advprobdist(\advguess_1^0) + \cdots + 
	\advprobdist(\advguess_n^0) \leq 1/k$.  \qed 
\end{proof}

\subsection{$k$-symmetry \emph{versus} $(k, \ell)$-anonymity}

As proven in Theorem~\ref{theo-k-symmetry}, $k$-symmetry provides 
protection against active 
attacks regardless of the number of sybil nodes inserted by the attacker, as 
opposed to $(k, \ell)$-anonymity which uses $\ell$ as a parameter on the 
maximum number of sybil nodes. In spite of that, $(k, \ell)$-anonymity is not 
weaker than $k$-symmetry. As we prove next, they are in fact incomparable. 

\begin{theorem}
Let $\mathcal{G}_{k, \ell}$ be the universe of anonymised graphs such that no 
adversary with $\ell$ sybil nodes or less can re-identify a victim with 
probability lower or equal than $1/k$. 
There exist $k > 1$ and graphs $G, G', G'' \in \mathcal{G}_{k, \ell}$ such that:
\begin{itemize}
	\item $G$ satisfies $k$-symmetry, but $G$ does not satisfy $(k, 
	\ell)$-anonymity for some $\ell \geq 1$. 
	\item $G'$ satisfies $(k, \ell)$-anonymity for some $\ell \geq 1$, but $G'$ 
	does not satisfy $k$-symmetry. 
	\item $G''$ neither satisfy $k$-symmetry nor $(k, \ell)$-anonymity for some 
	$\ell \geq 1$. 
\end{itemize} 
\end{theorem}
\begin{proof}
Figure~\ref{fig-example-a} shows a \mbox{$2$-symmetric} graph $G$ which, 
for $2\le\ell\le 8$, does not satisfy $(k,\ell)$-anonymity for any $k>1$. 
Moreover, Figure~\ref{fig-example-b} 
shows a $(2,1)$-anonymous graph $G'$ which can be verified 
not to satisfy $k$-symmetry for any $k>1$. In fact, 
this graph even fails to satisfy $k$-degree anonymity 
for any $k>1$. An example of a graph $G''$ proving the correctness of the last 
statement is displayed in Figure~\ref{fig-example-c}. That graph is neither 
$2$-symmetric nor $(2, 2)$-anonymous.\qed 
\end{proof}

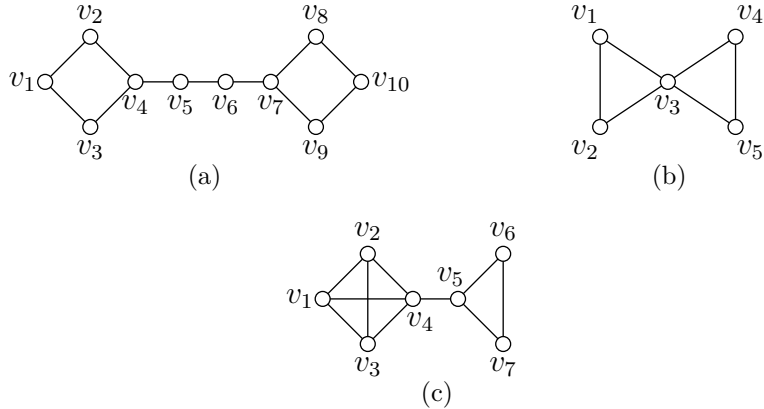
\begin{figure}[!th]

\begin{center}

\begin{tabular}{cc}
\subfigure[\label{fig-example-a}]{

\begin{tikzpicture}[inner sep=0.7mm, place/.style={circle,draw=black,
fill=white},gray1/.style={circle,draw=black!99, 
fill=black!75},gray2/.style={circle,draw=black!99, 
fill=black!25},transition/.style={rectangle,draw=black!50,fill=black!20,thick}, 
line width=.5pt]

\coordinate (v1) at (-2.1,0);
\coordinate (v2) at (-1.5,0.6);
\coordinate (v3) at (-1.5,-0.6);
\coordinate (v4) at (-0.9,0);
\coordinate (v5) at (-0.3,0);
\coordinate (v6) at (0.3,0);
\coordinate (v7) at (0.9,0);
\coordinate (v8) at (1.5,0.6);
\coordinate (v9) at (1.5,-0.6);
\coordinate (v10) at (2.1,0);

\draw[black] (v1) -- (v2) -- (v4) -- (v5) -- (v6) -- (v7) -- (v8) -- (v10) -- (v9) -- (v7); 
\draw[black] (v1) -- (v3) -- (v4);

\node[place] at (v1) {};
\node[place] at (v2) {};
\node[place] at (v3) {};
\node[place] at (v4) {};
\node[place] at (v5) {};
\node[place] at (v6) {};
\node[place] at (v7) {};
\node[place] at (v8) {};
\node[place] at (v9) {};
\node[place] at (v10) {};

\coordinate [label=center:{$v_1$}] (lv1) at (-2.4,0);
\coordinate [label=center:{$v_2$}] (lv2) at (-1.5,0.9);
\coordinate [label=center:{$v_3$}] (lv3) at (-1.5,-0.9);
\coordinate [label=center:{$v_4$}] (lv4) at (-0.9,-0.3);
\coordinate [label=center:{$v_5$}] (lv5) at (-0.3,-0.3);
\coordinate [label=center:{$v_6$}] (lv6) at (0.3,-0.3);
\coordinate [label=center:{$v_7$}] (lv7) at (0.9,-0.3);
\coordinate [label=center:{$v_8$}] (lv8) at (1.5,0.9);
\coordinate [label=center:{$v_9$}] (lv9) at (1.5,-0.9);
\coordinate [label=center:{$v_{10}$}] (lv10) at (2.5,0);

\end{tikzpicture}

}
&
\subfigure[\label{fig-example-b}]{

\begin{tikzpicture}[inner sep=0.7mm, place/.style={circle,draw=black,
fill=white},gray1/.style={circle,draw=black!99, 
fill=black!75},gray2/.style={circle,draw=black!99, 
fill=black!25},transition/.style={rectangle,draw=black!50,fill=black!20,thick}, 
line width=.5pt]

\coordinate (v1) at (-0.9,0.6);
\coordinate (v2) at (-0.9,-0.6);
\coordinate (v3) at (0,0);
\coordinate (v4) at (0.9,0.6);
\coordinate (v5) at (0.9,-0.6);

\draw[black] (v3) -- (v1) -- (v2) -- (v3) -- (v4) -- (v5) -- (v3); 

\node[place] at (v1) {};
\node[place] at (v2) {};
\node[place] at (v3) {};
\node[place] at (v4) {};
\node[place] at (v5) {};

\coordinate [label=center:{$v_1$}] (lv2) at (-1.1,0.9);
\coordinate [label=center:{$v_2$}] (lv3) at (-1.1,-0.9);
\coordinate [label=center:{$v_3$}] (lv4) at (0,-0.3);
\coordinate [label=center:{$v_4$}] (lv6) at (1.1,0.9);
\coordinate [label=center:{$v_5$}] (lv7) at (1.1,-0.9);

\coordinate [label=center:{\textcolor{white}{$v_0$}}] (lv0) at (-2.4,0);
\coordinate [label=center:{\textcolor{white}{$v_{10}$}}] (lv10) at (2.5,0);

\end{tikzpicture}

}\\

\multicolumn{2}{c}{

\subfigure[\label{fig-example-c}]{

\begin{tikzpicture}[inner sep=0.7mm, place/.style={circle,draw=black,
fill=white},gray1/.style={circle,draw=black!99, 
fill=black!75},gray2/.style={circle,draw=black!99, 
fill=black!25},transition/.style={rectangle,draw=black!50,fill=black!20,thick}, 
line width=.5pt]

\coordinate (v1) at (-1.5,0);
\coordinate (v2) at (-0.9,0.6);
\coordinate (v3) at (-0.9,-0.6);
\coordinate (v4) at (-0.3,0);
\coordinate (v5) at (0.3,0);
\coordinate (v6) at (0.9,0.6);
\coordinate (v7) at (0.9,-0.6);

\draw[black] (v1) -- (v2) -- (v4) -- (v5) -- (v6) -- (v7) -- (v5); 
\draw[black] (v4) -- (v3) -- (v1) -- (v4);
\draw[black] (v2) -- (v3);

\node[place] at (v1) {};
\node[place] at (v2) {};
\node[place] at (v3) {};
\node[place] at (v4) {};
\node[place] at (v5) {};
\node[place] at (v6) {};
\node[place] at (v7) {};

\coordinate [label=center:{$v_1$}] (lv1) at (-1.8,0);
\coordinate [label=center:{$v_2$}] (lv2) at (-0.9,0.9);
\coordinate [label=center:{$v_3$}] (lv3) at (-0.9,-0.9);
\coordinate [label=center:{$v_4$}] (lv4) at (-0.2,-0.3);
\coordinate [label=center:{$v_5$}] (lv5) at (0.2,0.3);
\coordinate [label=center:{$v_6$}] (lv6) at (0.9,0.9);
\coordinate [label=center:{$v_7$}] (lv7) at (0.9,-0.9);

\coordinate [label=center:{\textcolor{white}{$v_0$}}] (lv0) at (-2.4,0);
\coordinate [label=center:{\textcolor{white}{$v_{10}$}}] (lv10) at (2.5,0);

\end{tikzpicture}

}
}
\end{tabular}

\end{center}

\caption{Example graphs. (a) A $2$-symmetric graph not satisfying 
$(k,\ell)$-anonymity for $k>1$ and $2\le \ell\le 8$. (b) A $(2,1)$-anonymous graph  
not satisfying $k$-symmetry for $k>1$. (c) A graph where the success probability 
of any active attack leveraging $2$ sybil nodes is at most $1/2$, 
despite the graph neither satisfying $(2,2)$-anonymity nor $2$-symmetry.}
\label{fig-examples}
\end{figure}

Of independent interest is our proof that $k$-automorphism~\cite{ZCO2009} does 
not protect against active attacks. This is a surprising result, given that 
$k$-automorphism and $k$-symmetry have traditionally been
considered equivalent. 
We refer the interested reader to Appendix~\ref{app-automorphism}. 
 
\section{Algorithm \textsc{K-Match} guarantees $k$-symmetry}
\label{sec-k-match}

In this section we prove that the algorithm 
\textsc{K-Match}, proposed in~\cite{ZCO2009}  
as a sufficient condition to achieve $k$-automorphism, 
also guarantees $k$-symmetry. 
Given a graph $G$ and a value of $k$, 
the \textsc{K-Match} algorithm obtains a supergraph 
$G'$ of $G$ satisfying the following conditions:
\begin{enumerate}
\item $V_{G'}\supseteq V_G$ and $E_{G'}\supseteq E_G$.
\item There exist $k-1$ automorphisms $\gamma_1, \gamma_2, \ldots, \gamma_{k-1}$ of $G'$ such that:
\begin{enumerate}
\item For every $v\in V_{G'}$ and every $i\in\{1,\ldots,k-1\}$, $\gamma_i(v)\neq v$.
\item For every $v\in V_{G'}$ and every $i,j\in\{1,\ldots,k-1\}$, $i\neq j \Longleftrightarrow \gamma_i(v)\neq\gamma_j(v)$.
\item For every $v\in V_{G'}$ and every $i,j$ such that $1\le i<j\le k-1$, $\gamma_{i+j}(v)=\gamma_i(\gamma_j(v))=\gamma_j(\gamma_i(v))$, with addition taken modulo $k$.
\end{enumerate}
\end{enumerate}

To obtain $G'$, the algorithm first splits the vertices 
of $G'$ into $k$ groups and arranges them 
in a $k$-column 
matrix $M$ called the \emph{vertex alignment table} 
(\emph{VAT} for short). If $|V_G|$ is not a multiple of $k$, a number 
of dummy vertices is added to achieve this property.
The VAT is organised in such a manner that the number 
of graph editions to perform in the second step of the process 
is close to the minimum. For convenience, in what follows we will 
denote by $v_{ij}$ the vertex of $G'$ placed 
in position $M_{ij}$ of the VAT. 
The second step of the method consists in 
adding edges to $E_{G'}$ in such a way that conditions 2.a to 2.c 
are enforced. To that end, for every edge 
$(v_{ij},v_{pq})$, all edges of the form $(v_{i,j+t},v_{p,q+t})$, 
additions modulo $k$, are added to $E_{G'}$ if they did not previously exist. 
 
Figure~\ref{fig-ex-vat} shows an example of a VAT allowing 
to enforce $3$-automorphism on the graph 
of Figure~\ref{fig-sybils}\footnote{This table is not necessarily 
the one created by the first step of \textsc{K-Match}, but it serves to illustrate 
the second step, which is the one that guarantees 
the privacy property and will be the basis of the main result in this section.}. 
This VAT encodes two functions $f_1,f_2\colon V_{G'}\rightarrow V_{G'}$: 
$$f_1=\{(1,F),(F,D),(D,1),(C,A),(A,B),(B,C),(2,3),(3,E),(E,2)\},$$
that is, a function such that the image of every element is the one located 
one column to its right (modulo $3$) on the same row, and 
$$f_2=\{(1,D),(F,1),(D,F),(C,B),(A,C),(B,A),(2,E),(3,2),(E,3)\},$$
that is, a function such that the image of every element is the one located 
two columns to its right (modulo $3$) on the same row. 

\begin{figure}[h]
\begin{center}
\begin{tabular}{|c|c|c|}
\hline
\quad 1\quad$$&\quad F\quad$$&\quad D\quad$$\\
\hline
\quad C\quad$$&\quad A\quad$$&\quad B\quad$$\\
\hline
\quad 2\quad$$&\quad 3\quad$$&\quad E\quad$$\\
\hline
\end{tabular}
\end{center}
\caption{An example of a VAT for the graph shown in Figure~\ref{fig-sybils}.} 
\label{fig-ex-vat}
\end{figure}
 
In general, these functions are not automorphisms of $G'$ upon creation of the VAT. 
It is the second step of the method that will transform them into automorphisms 
by performing all necessary edge-copying operations. 
For example, the edge $(C,A)$ needs to be added to $G'$ because 
$(A,B)\in E_G$ but $(f_2(A),f_2(B))=(C,A)\notin E_G$; and $(A,3)$ 
needs to be added because $(B,E)\in E_G$ but $(f_2(B),f_2(E))=(A,3)\notin E_G$. 
Once the method is executed, each automorphism $\gamma_t$, $t\in\{1,\ldots,k-1\}$, 
defined in item~2 above is completely 
specified by the VAT, as $\gamma_t(v_{ij})=v_{i,j+t}$, with addition modulo $k$, 
for every $i\in\left\{1,\ldots,\left\lceil\frac{|V_G|}{k}\right\rceil\right\}$ 
and every $j\in\{1,\ldots,k\}$.  

We now show the link between the \textsc{K-Match} method and $k$-symmetry. 

\begin{theorem}\label{prop-km-k-sym}
Let $G=(V,E)$ be a graph and let $G'=(V',E')$ the result of applying 
algorithm \textsc{K-Match} to $G$ for some parameter $k$. 
Then, $G'$ satisfies $k$-symmetry. 
\end{theorem}

\begin{proof}
Let $u\in V_{G'}$ be an arbitrary vertex of $G'$, 
and let $v_1=\gamma_1(u)$, $v_2=\gamma_2(u)$, \ldots, $v_{k-1}=\gamma_{k-1}(u)$ 
be the images of $u$ by the automorphisms $\gamma_1, \gamma_2, \ldots, \gamma_{k-1}$ 
enforced on $G'$ by the execution of \textsc{K-Match}. 
By definition, we have that $u \cong v_1 \cong v_2 \cong \ldots \cong v_{k-1}$ 
and, by conditions 2.a and 2.b, they are pairwise different. 
Thus, $\left\vert[u]_{\cong}\right\vert=k$, hence $G'$ is $k$-symmetric. \qed
\end{proof}

The most relevant consequence of Theorem~\ref{prop-km-k-sym} 
is that algorithm \textsc{K-Match} can also be used for protecting 
graphs against active adversaries, as it will ensure
that no victim is re-identified with probability greater than $1/k$.

\section{Experiments}\label{sec-experiments} 

The purpose of these experiments\footnote{We performed our experiments 
on the HPC platform of the University of Luxembourg~\cite{VBCG_HPCS14}. 
In particular, we ran our experiments on the Gaia cluster of the UL HPC. 
A detailed description of the Gaia cluster is available 
at \url{https://hpc.uni.lu/systems/gaia/}. 
The implementations of the graph generators, 
anonymisation methods and attack simulations are available 
at \url{https://github.com/rolandotr/graph}.} 
is to demonstrate the effectiveness and usability of $k$-symmetry, 
enforced using the \textsc{K-Match} algorithm, for protecting graphs 
against active adversaries leveraging a sufficiently large number 
of sybil nodes and the strongest attack strategy reported in the literature, 
namely the robust active attack introduced in~\cite{MRT2018robsyb}. 
Effectiveness is assessed in terms of the success rate measure 
used in previous works on active 
attacks~\cite{MauwTrujilloXuan2016,MauwRamirezTrujillo2016,MRT2018,MRT2018robsyb}, 
whereas usability is assessed in terms of several structural utility measures. 
In what follows, 
we describe the experimental setting, display the empirical results obtained 
and conclude the section with a discussion of these results. 

\subsection{Experimental setting}\label{subsec-experimetal-setting}  

In order to make the results reported in this section comparable 
to previous works on active attacks and countermeasures 
against them~\cite{MRT2018,MRT2018robsyb}, 
we study the behaviour of our proposed method on two types 
of randomly generated synthetic graphs. 
In the first case, we use Erd\H{o}s-R\'enyi (ER) random graphs \cite{Erdos1959}. 
We generated $200,000$ ER graphs, $10,000$ for each density value 
in the set $\{0.1, 0.15, \ldots, 0.95, 1.0\}$. 
The second group of synthetic graphs was generated according to 
the Barab\'{a}si-Albert (BA) model~\cite{Barabasi1999}, 
which generates scale-free graphs. 
We used seed graphs of order $50$ and every graph was grown 
by adding $150$ vertices, and performing the corresponding edge additions. 
The BA model has a parameter $m$ defining the number of new edges added 
for every new vertex. We generated $10,000$ graphs for every value of $m$ 
in the set $\{5, 10, \ldots, 50\}$. In generating each graph, 
the type of the seed graph was randomly selected among the following choices: 
a complete graph, an $m$-regular ring lattice, or an ER random graph 
of density~$0.5$. The probability of selecting each choice  
was set as $\sfrac{1}{3}$. 
In both cases, the generated synthetic graphs have $200$ nodes. 
Based on the discussion on the plausible number 
of sybil nodes in Section~\ref{sec-adversarial-model}, 
we make the number of sybils $\ell=\lceil\log_2 200\rceil=8$. 

We analyse three values for the privacy parameter $k$: 
a low value, $k=2$; a high value, $k=\ell=8$; and an intermediate value,
$k=5$. For every value of $k$, we compare 
the behaviour of the \mbox{\textsc{K-Match}} algorithm, which ensures $k$-symmetry, 
and Mauw et al.'s algorithm for enforcing \mbox{$(k,\Gamma_{G,1})$-adjacency} 
anonymity~\cite{MRT2018}. In order to build the vertex alignment table, 
algorithm \textsc{K-Match} requires the vertex set 
of the input graph to be partitioned into $k$ subsets such that 
the number of edges linking vertices in different subsets is close to the minimum. 
We used the multilevel $k$-way partitioning method reported in~\cite{Karypis1998}, 
in specific its implementation included in the METIS 
library~\footnote{Available at \url{http://glaros.dtc.umn.edu/gkhome/views/metis}.}, 
for efficiently obtaining such a partition. 
The effectiveness of the anonymisation methods is measured in terms 
of their resistance to the robust active attack described 
in~\cite{MRT2018robsyb}. Thus, following the attacker-defender game 
described in Section~\ref{sec-adversarial-model}, for every graph 
we first run the attacker subgraph 
creation stage. Then, for every resulting graph, 
we obtain the two variants of anonymised graphs. Finally, for each perturbed graph, 
we simulate the execution of the re-identification stage and compute 
its success rate as defined in~\cite{MRT2018robsyb}, that is  

\begin{equation}\label{eq-succ-rate-robust}
SuccessRate=\left\{\begin{array}{ll}
\frac{\sum_{X\in \mathcal{X}} 
p_{X}}{|\mathcal{X}|}
&\quad\text{if }\mathcal{X}\neq\emptyset\\
0&\quad\text{otherwise}
\end{array}
\right. 
\end{equation}

\noindent
where $\mathcal{X}$ is the set of equally-most-likely sybil subgraphs 
retrieved in $\transfg$ by the third phase of the attack, and 

\begin{displaymath}
p_{X} =\left\{
\begin{array}{ll}
\frac{1}{|\mathcal{Y}_{X}|} \quad & \quad \text{if} 
\quad Y \in \mathcal{Y}_{X}\\
                  0 \quad & \quad \text{otherwise}
\end{array}
\right.
\end{displaymath}

\noindent
with $\mathcal{Y}_{X}$ containing all equally-most-likely 
fingerprint matchings according to $X$. 
In order to obtain the scores used for comparing the effectiveness 
of the different approaches, we compute for every method 
the average of the success rates over every group of $10,000$ graphs sharing 
the same set of parameter choices. 

The anonymisation methods are also compared in terms 
of utility. To that end, we measure the distortion caused 
by each method on a number of global graph statistics, namely the global 
clustering coefficient, the averaged local clustering coefficient, 
and the similarity between the degree distributions, 
measured in terms of the cosine of the angle between the degree vectors, 
following the approach introduced in~\cite{SecGraph,MauwRamirezTrujillo2016}. 

\subsection{Results and discussion}\label{subsec-discussion}

Figure~\ref{fig-succ-prob-er-ba-200-8-syb-at-1} shows the success rates 
of the attack on both random graph collections, whereas 
Figures~\ref{fig-dd-er-ba-200-8-syb-at-1}, 
\ref{fig-gcc-er-ba-200-8-syb-at-1} and~\ref{fig-avg-lcc-er-ba-200-8-syb-at-1} 
show utility values in terms of degree distribution similarity, 
variation of global clustering coefficient 
and variation of averaged local clustering coefficient, respectively. 

Regarding the effectiveness of the anonymisation methods, 
the results in Figure~\ref{fig-succ-prob-er-ba-200-8-syb-at-1} clearly show 
that \textsc{K-Match} is considerably more effective 
against the robust active attack than $(k,\Gamma_{G,1})$-adjacency anonymity. 
These results are particularly relevant in light of the fact 
that $(k,\Gamma_{G,1})$-adjacency anonymity was until now 
the sole formal privacy property to provide non-negligible protection 
against the original active attack and some instances 
of the robust active attack~\cite{MRT2018,MRT2018robsyb}. 
Finally, the experiments shown here are the first ones 
where the robust active attack leveraging $\lceil \log_2 n \rceil$ sybils 
is shown to be consistently thwarted by anonymisation methods 
based on formal privacy properties. So far, this had only been achieved 
via the addition of (large amounts of) random noise~\cite{MRT2018robsyb}. 

Regarding utility, both methods have a small impact on the overall similarities 
of the degree distributions, 
as illustrated in Figure~\ref{fig-dd-er-ba-200-8-syb-at-1}. 
This does not mean that the degrees are not affected 
by the methods. In fact, both methods make most degrees increase, but in a manner 
that does not significantly affect the ordering of vertices in terms 
of their degrees. Regarding clustering coefficient-based utilities, we can observe 
in Figures~\ref{fig-gcc-er-ba-200-8-syb-at-1} 
and~\ref{fig-avg-lcc-er-ba-200-8-syb-at-1} that the superior 
effectiveness of \textsc{K-Match} does come at the price of a larger degradation 
in the values of local and global clustering coefficients. 

The main takeaway from the experimental results presented in this section 
is that our refinement of the notion of re-identification probability 
for active adversaries has led to identifying, for the first time, 
an anonymisation method satisfying two key properties: 
(i) featuring a theoretically sound 
privacy guarantee against active attackers, 
and (ii) having this privacy guarantee translate 
into effective resistance to the strongest active attack reported so far, 
even when the attacker leverages a large number of sybil nodes.

\section{Conclusions}\label{sec-concl}

We have introduced a new probabilistic interpretation 
of active re-identification attacks on social graphs. This enables 
the privacy-preserving publication of social graphs in the presence 
of active adversaries by jointly preventing 
the attacker from unambiguously retrieving the set of sybil nodes, 
and from using the sybil nodes for re-identifying the victims. 
Under the new formulation, we have shown that the privacy property $k$-symmetry
provides a sufficient condition for the protection 
against active re-identification attacks. Moreover, we have shown  
that a previously existing efficient algorithm, \mbox{\textsc{K-Match}}, 
provides a sufficient condition for ensuring $k$-symmetry. 
Through a series of experiments, we have demonstrated that our approach allows, 
for the first time, to publish anonymised social graphs 
with formal privacy guarantees that effectively resist the robust active attack 
introduced in~\cite{MRT2018robsyb}, which is the strongest active re-identification 
attack reported in the literature, even when it leverages 
a large number of sybil nodes. 

The active adversary model addressed in this paper assumes 
that the (inherently dynamic) social graph is published only once. 
A more general scenario, where snapshots of a dynamic social network 
are periodically published in the presence of active adversaries, 
has recently been proposed in~\cite{CKMR2020}, and the robust active attack 
from~\cite{MRT2018robsyb} has been adapted to benefit from this scenario. 
Our main direction for future work consists in leveraging our methodology 
to propose anonymisation methods suited for this new publication scenario.  

\vspace{2mm}
\noindent 
\textbf{Acknowledgements:} The work reported in this paper
received funding from Luxembourg's Fonds National de la Recherche (FNR),
via grant C17/IS/11685812 (\mbox{PrivDA}). Part of this work was conducted 
while Yunior Ram\'irez-Cruz was visiting the School of Information Technology 
at Deakin University.  

\begin{figure}[H]
\centering
\subfigure[t][ER graphs, $k=2$]{
\includegraphics[scale=.54]{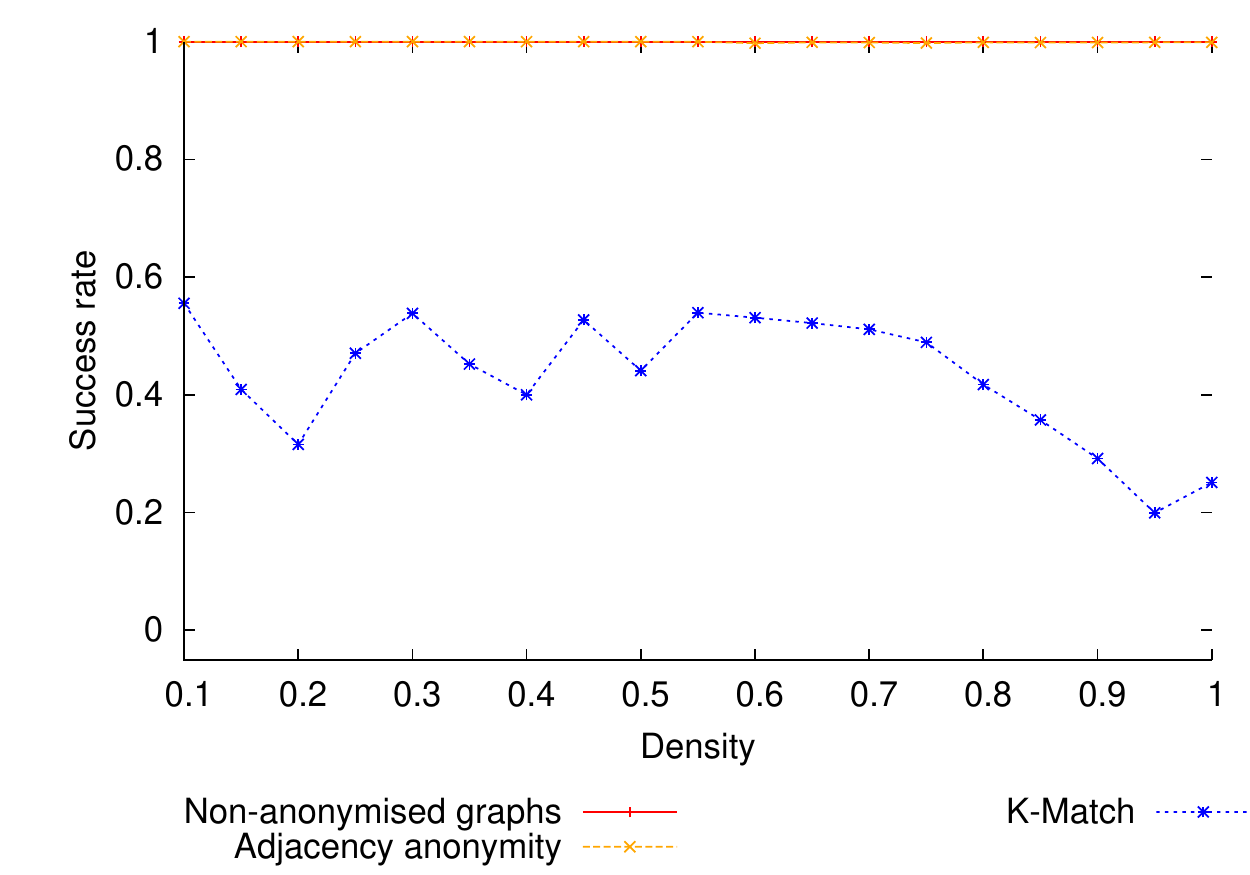}
}
\subfigure[t][BA graphs, $k=2$]{
\includegraphics[scale=.54]{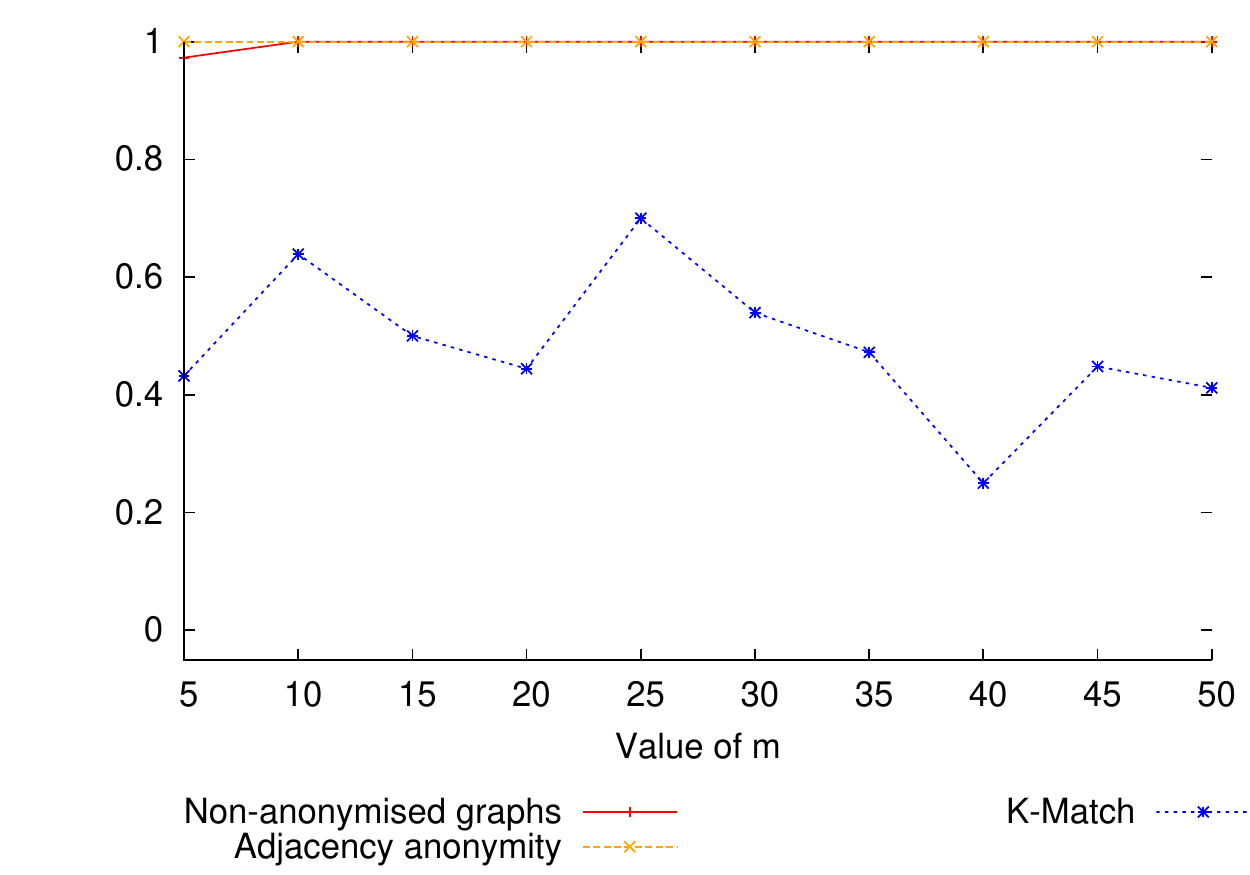}
}
\subfigure[t][ER graphs, $k=5$]{
\includegraphics[scale=.54]{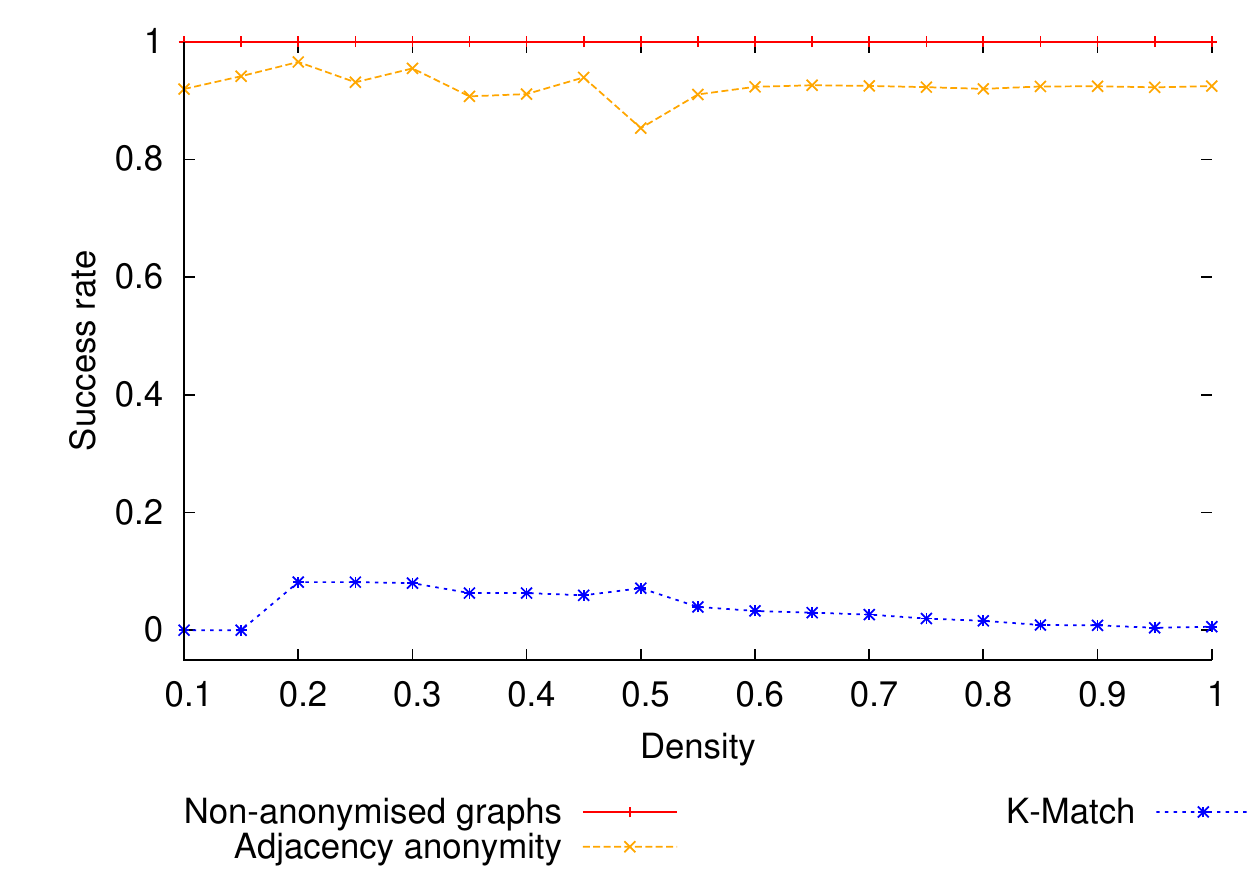}
}
\subfigure[t][BA graphs, $k=5$]{
\includegraphics[scale=.54]{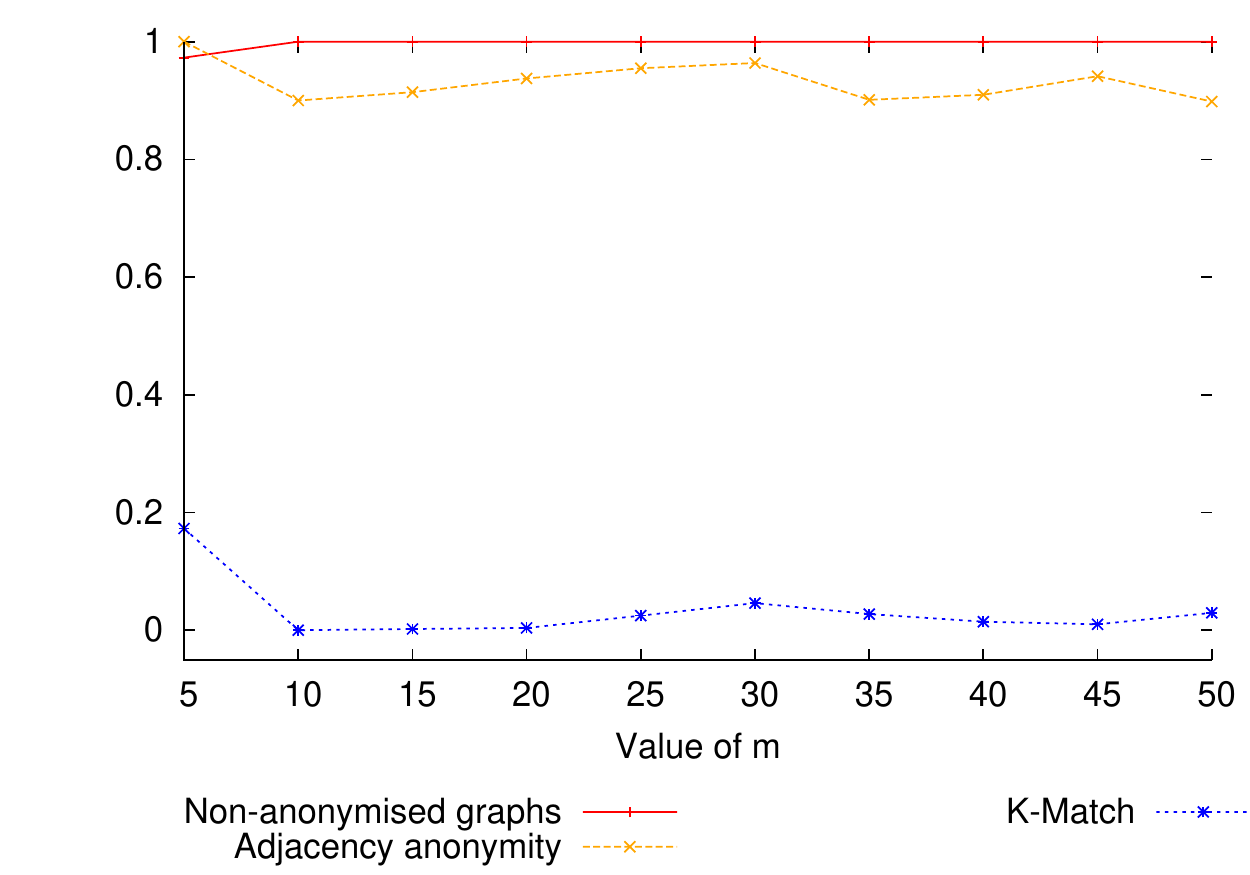}
}
\subfigure[t][ER graphs, $k=8$]{
\includegraphics[scale=.54]{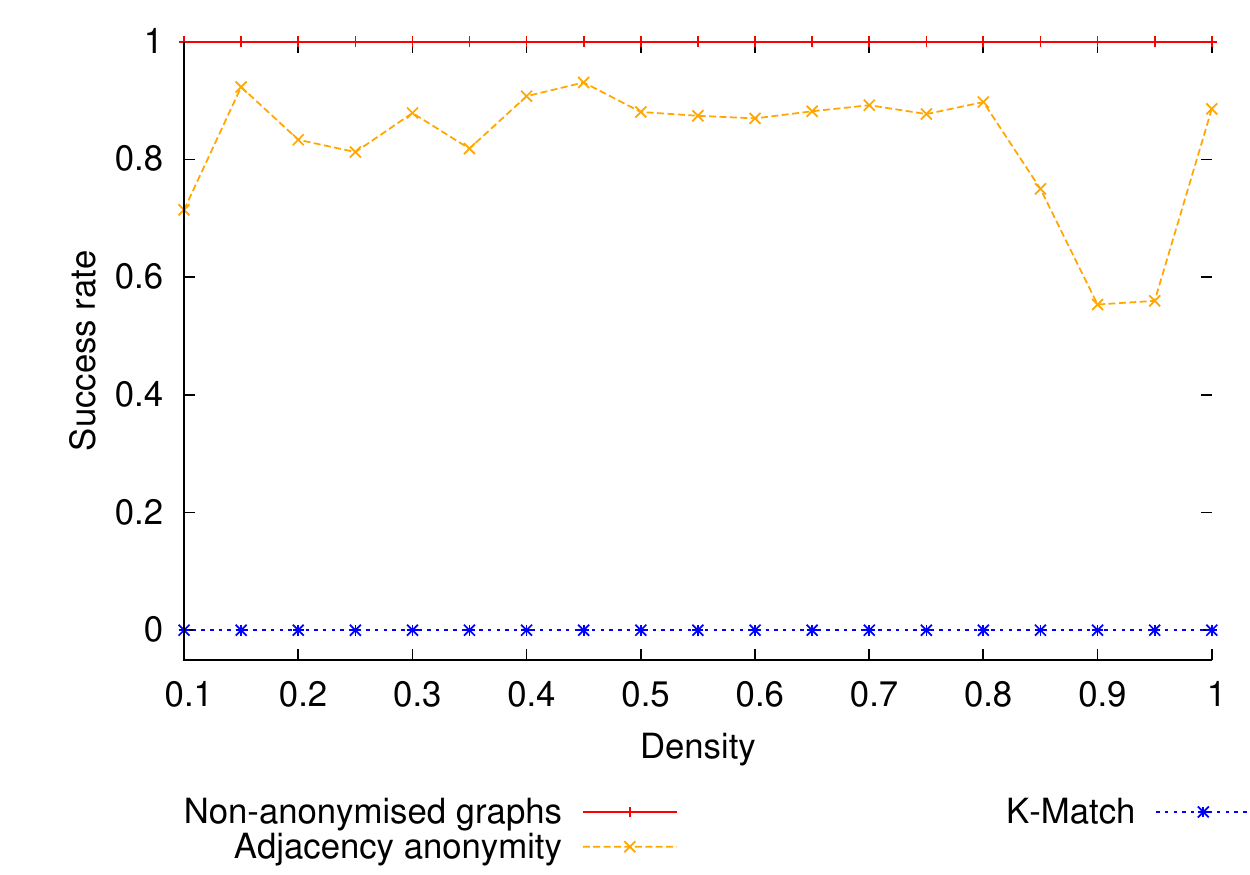}
}
\subfigure[t][BA graphs, $k=8$]{
\includegraphics[scale=.54]{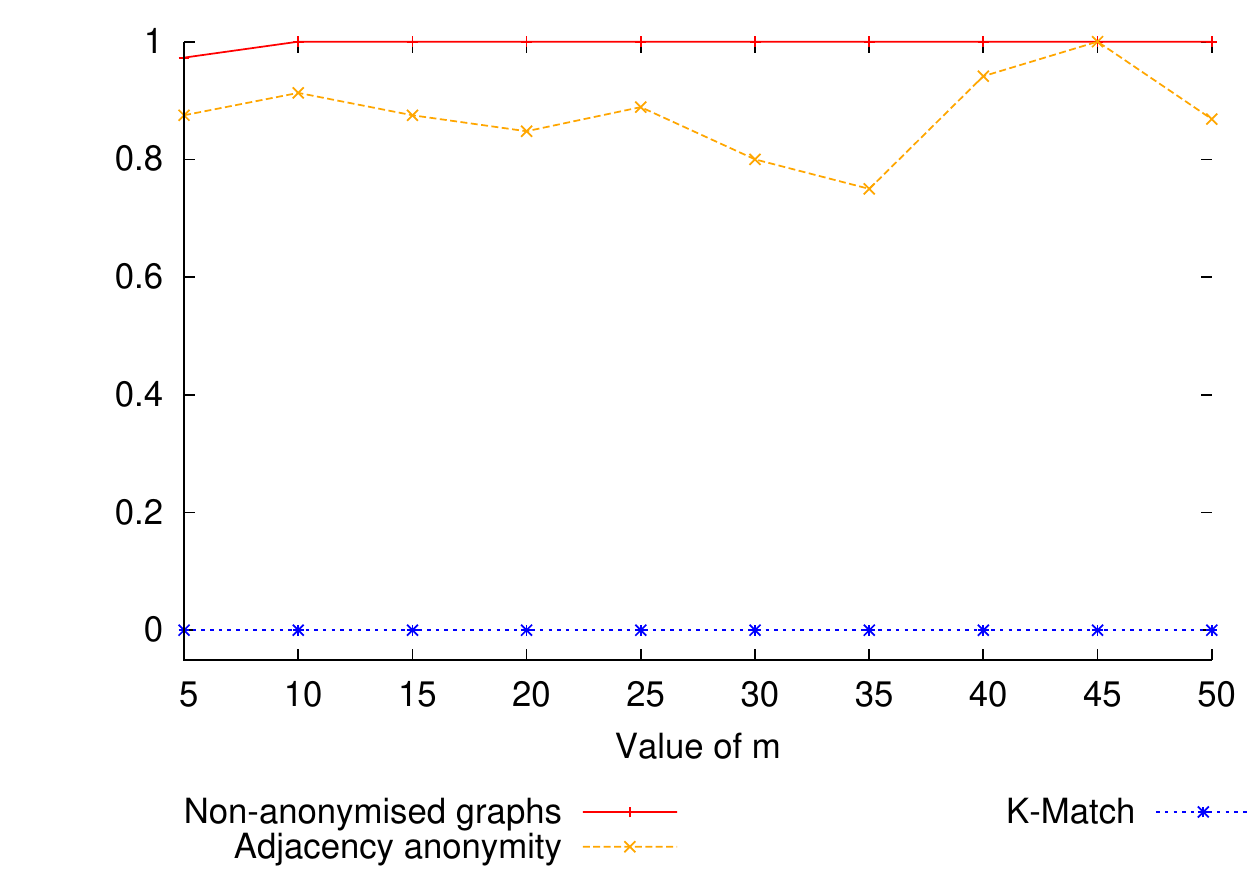}
}
\caption{Success rates of the robust active attack on the collections 
of Erd\H{o}s-R\'enyi (left) and Barab\'asi-Albert (right) random graphs, 
with $\ell=8$ and $k\in\{2,5,8\}$.} 
\label{fig-succ-prob-er-ba-200-8-syb-at-1} 
\end{figure} 



\begin{figure}[H]
\centering
\subfigure[t][ER graphs, $k=2$]{
\includegraphics[scale=.54]{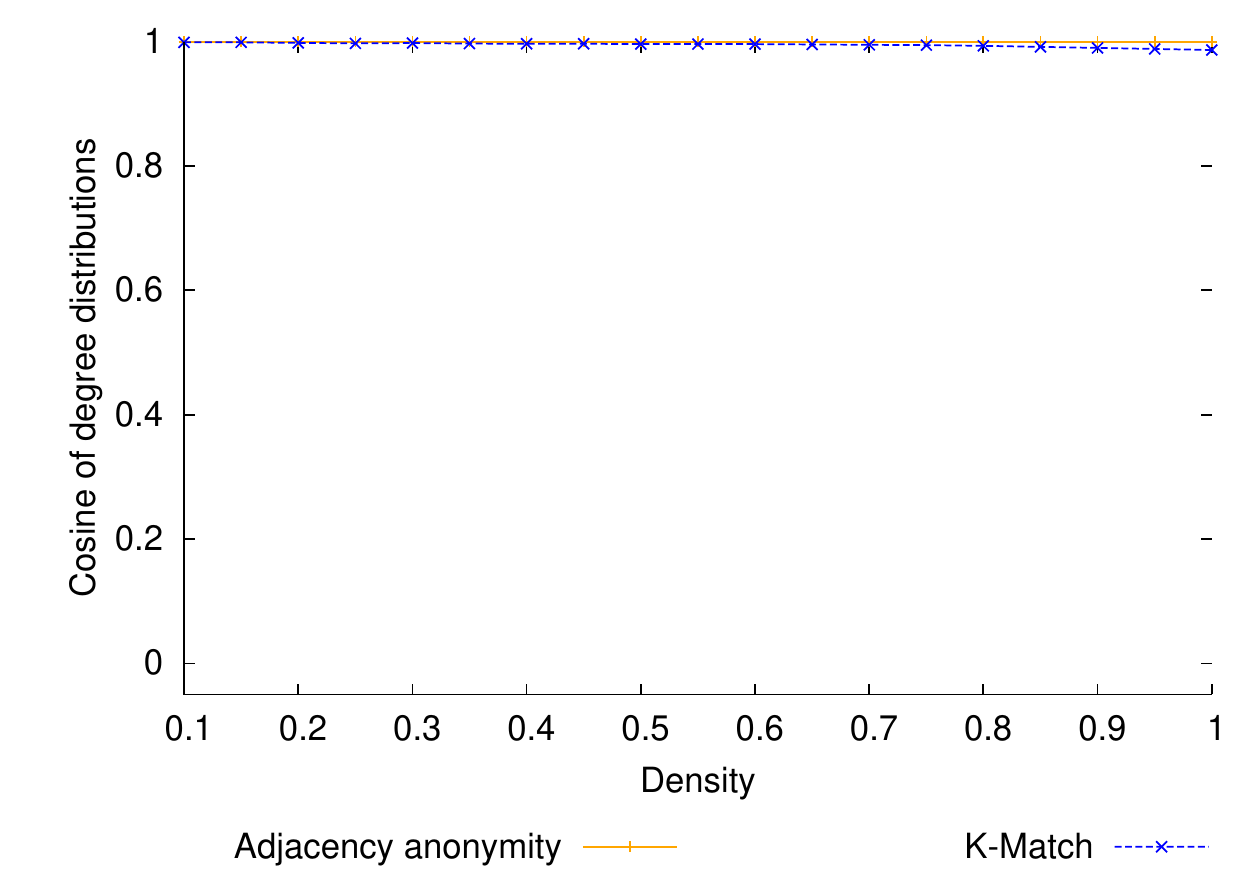}
}
\subfigure[t][BA graphs, $k=2$]{
\includegraphics[scale=.54]{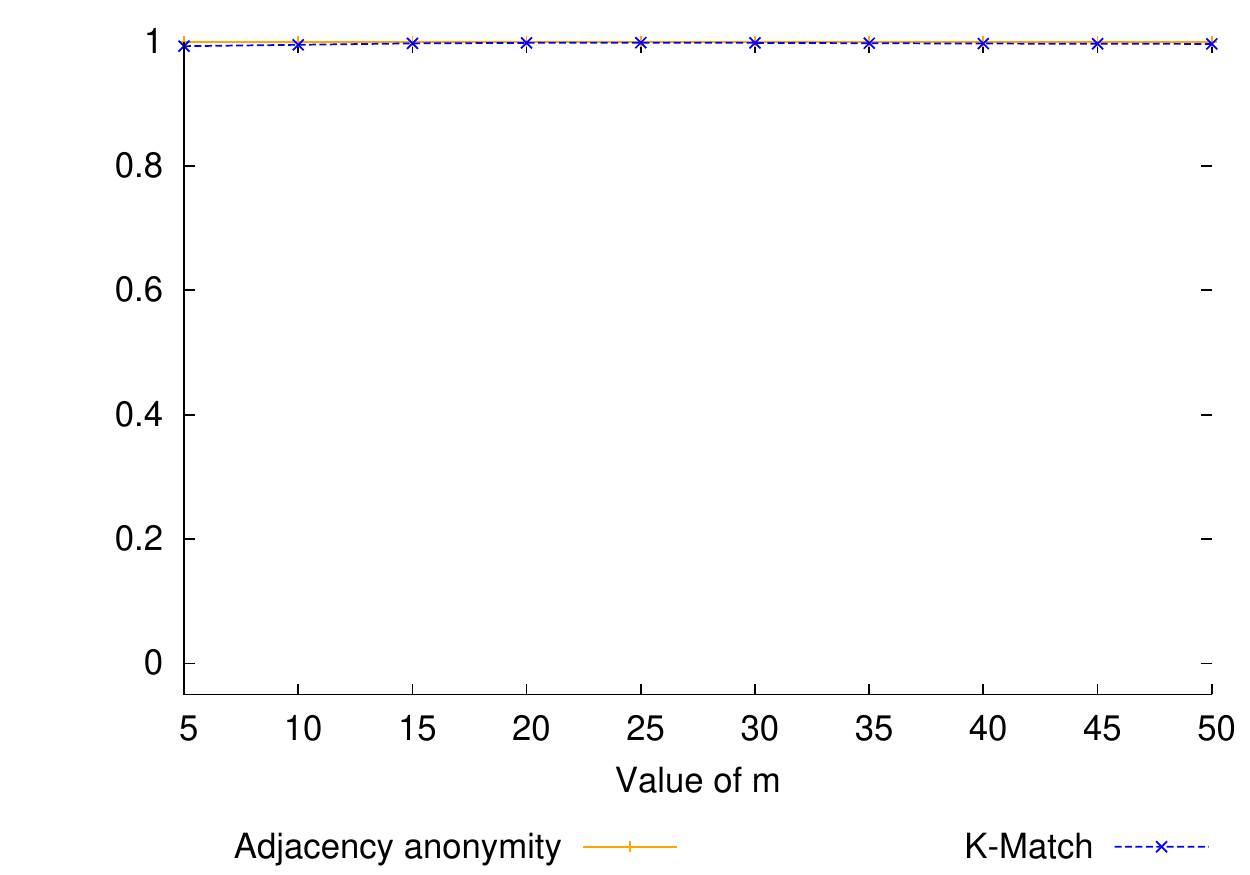}
}
\subfigure[t][ER graphs, $k=5$]{
\includegraphics[scale=.54]{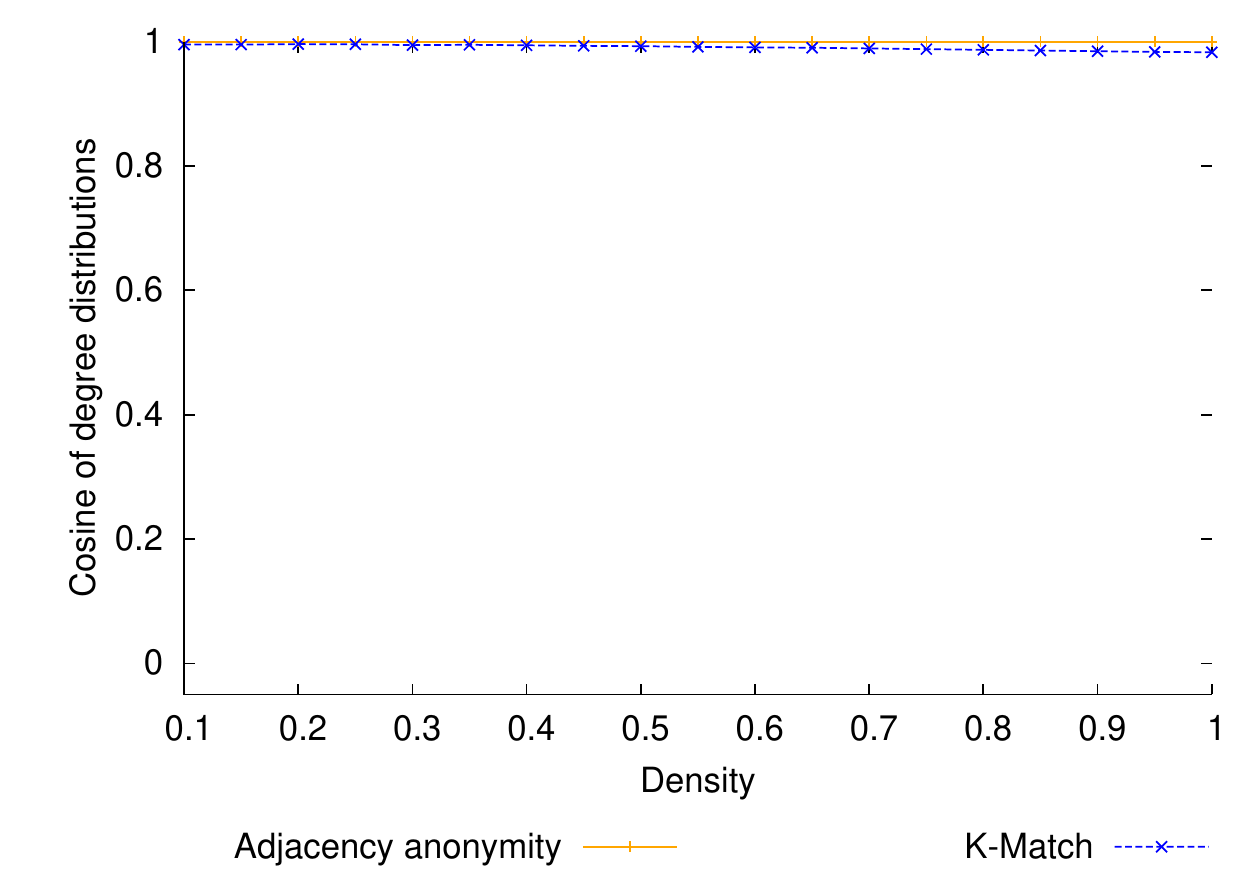}
}
\subfigure[t][BA graphs, $k=5$]{
\includegraphics[scale=.54]{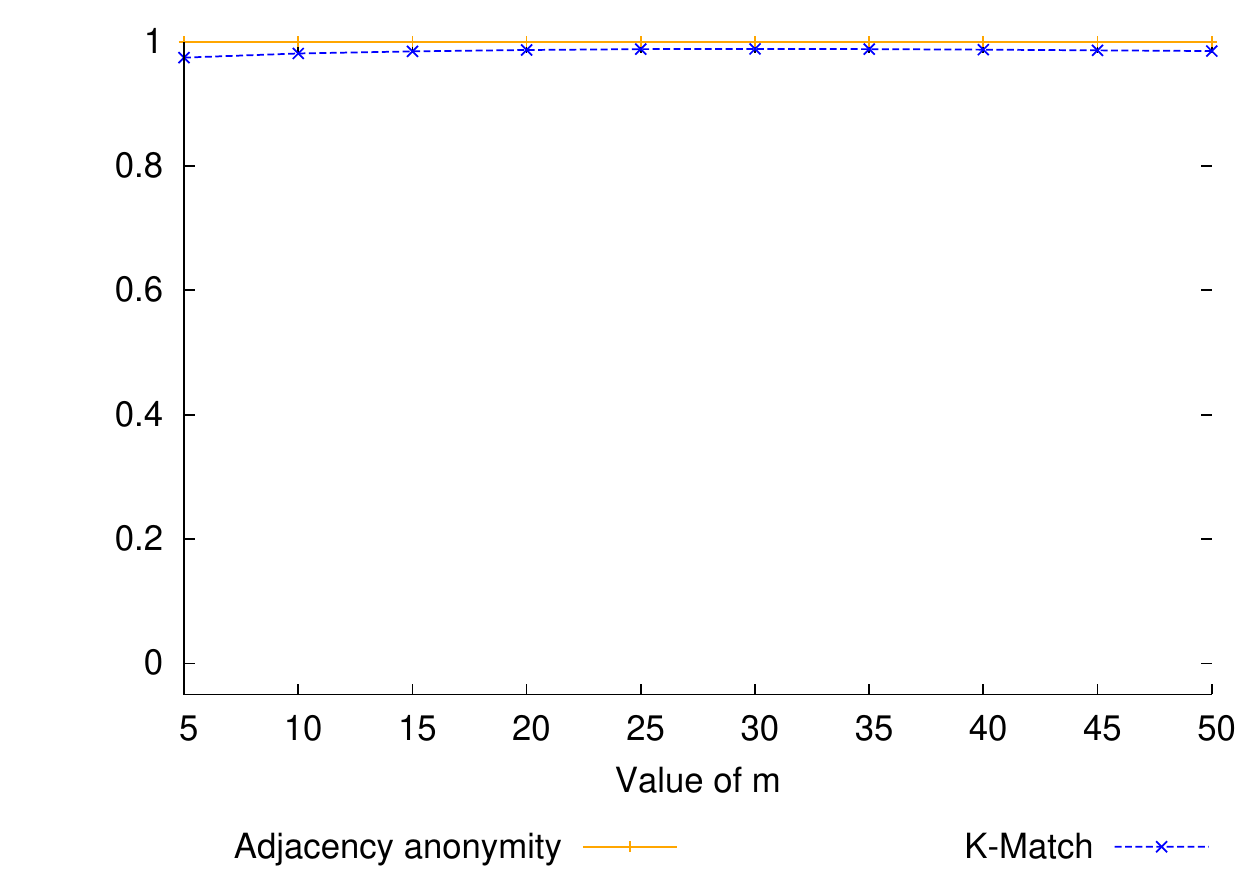}
}
\subfigure[t][ER graphs, $k=8$]{
\includegraphics[scale=.54]{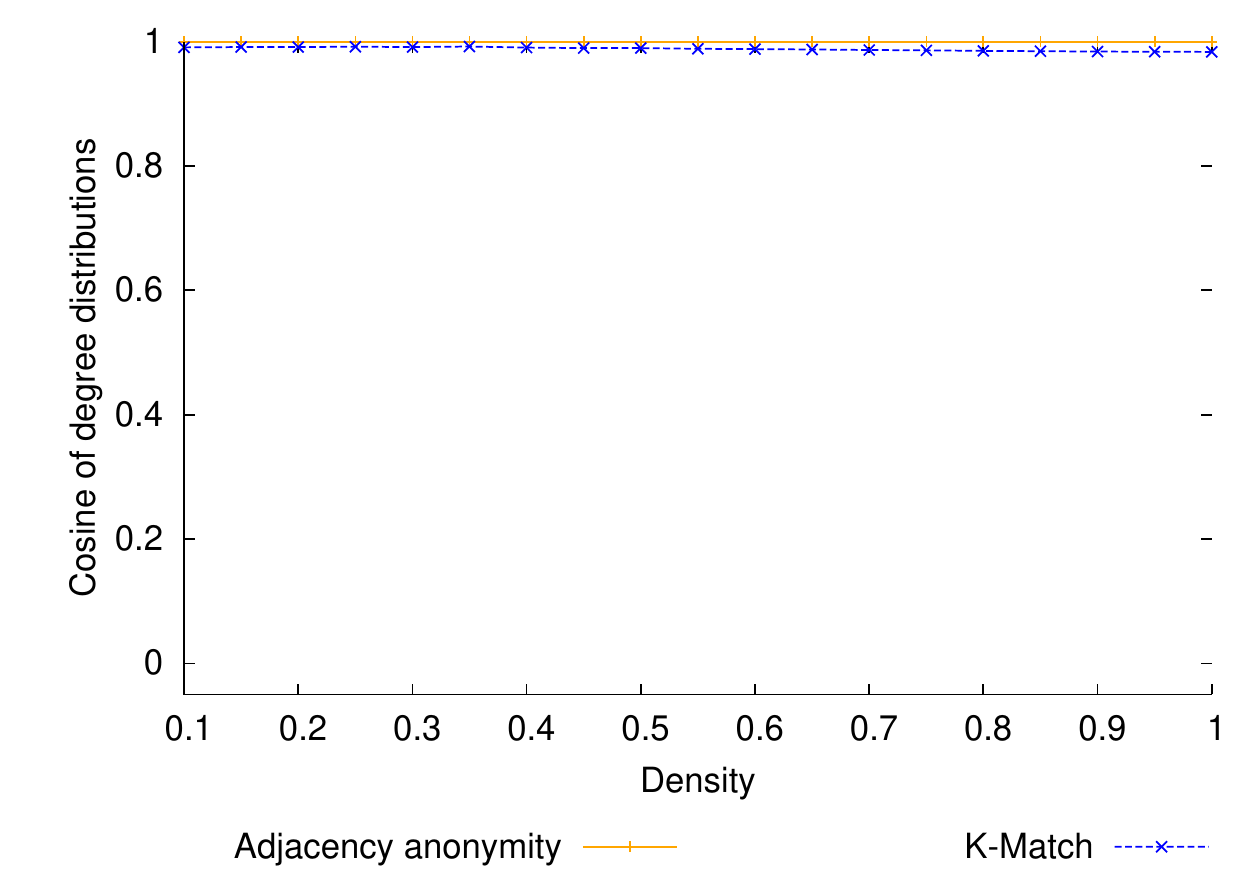}
}
\subfigure[t][BA graphs, $k=8$]{
\includegraphics[scale=.54]{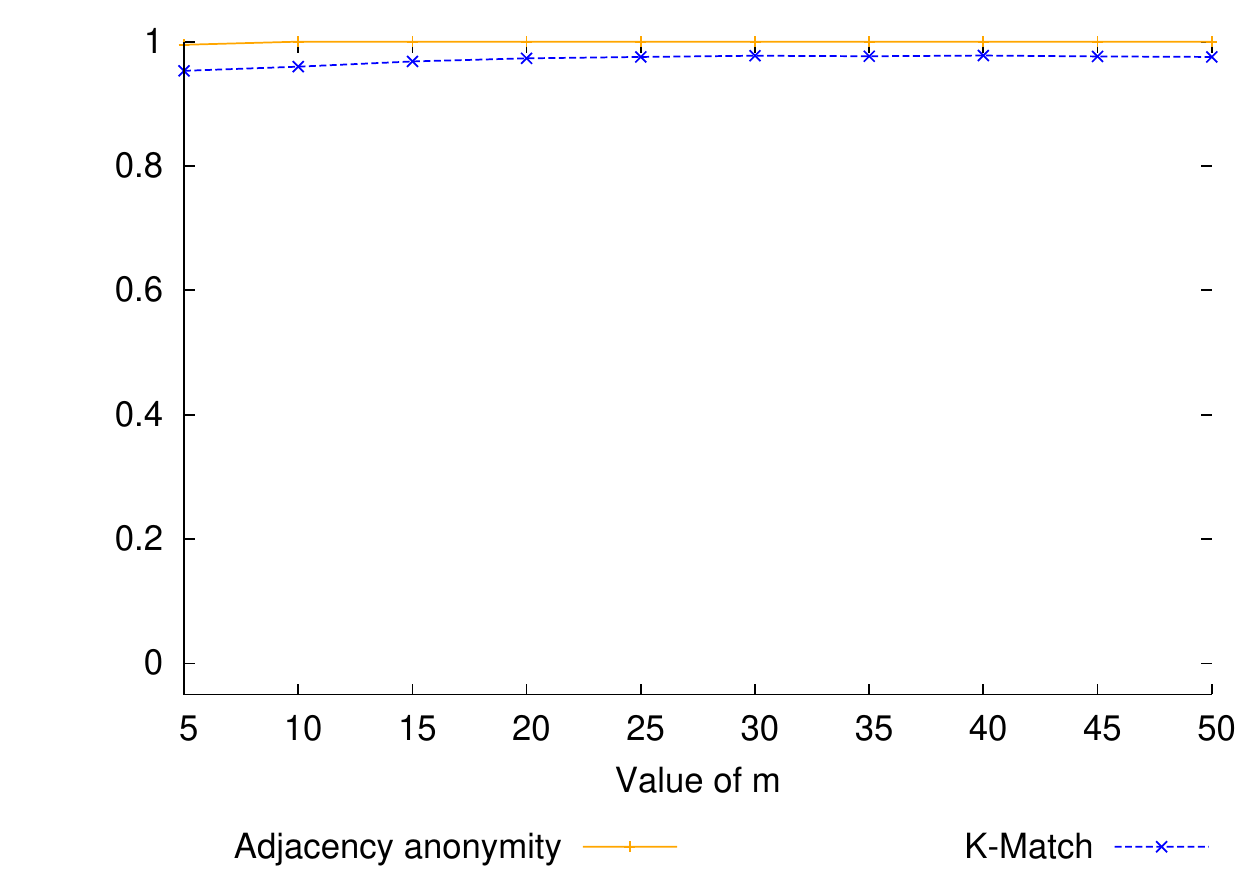}
}
\caption{Degree distribution similarities on the collections 
of Erd\H{o}s-R\'enyi (left) and Barab\'asi-Albert (right) random graphs, 
with $\ell=8$ and $k\in\{2,5,8\}$.} 
\label{fig-dd-er-ba-200-8-syb-at-1} 
\end{figure} 



\begin{figure}[H]
\centering
\subfigure[t][ER graphs, $k=2$]{
\includegraphics[scale=.54]{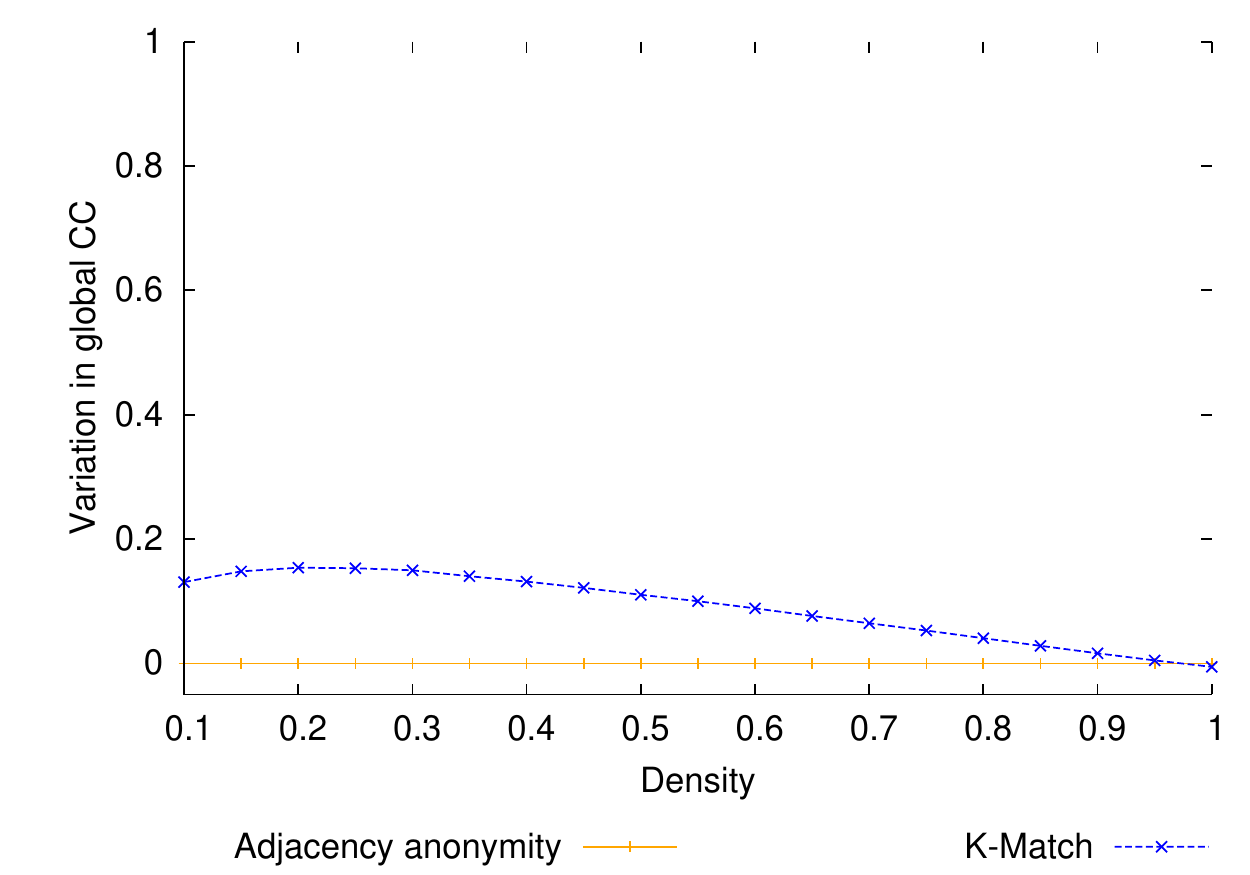}
}
\subfigure[t][BA graphs, $k=2$]{
\includegraphics[scale=.54]{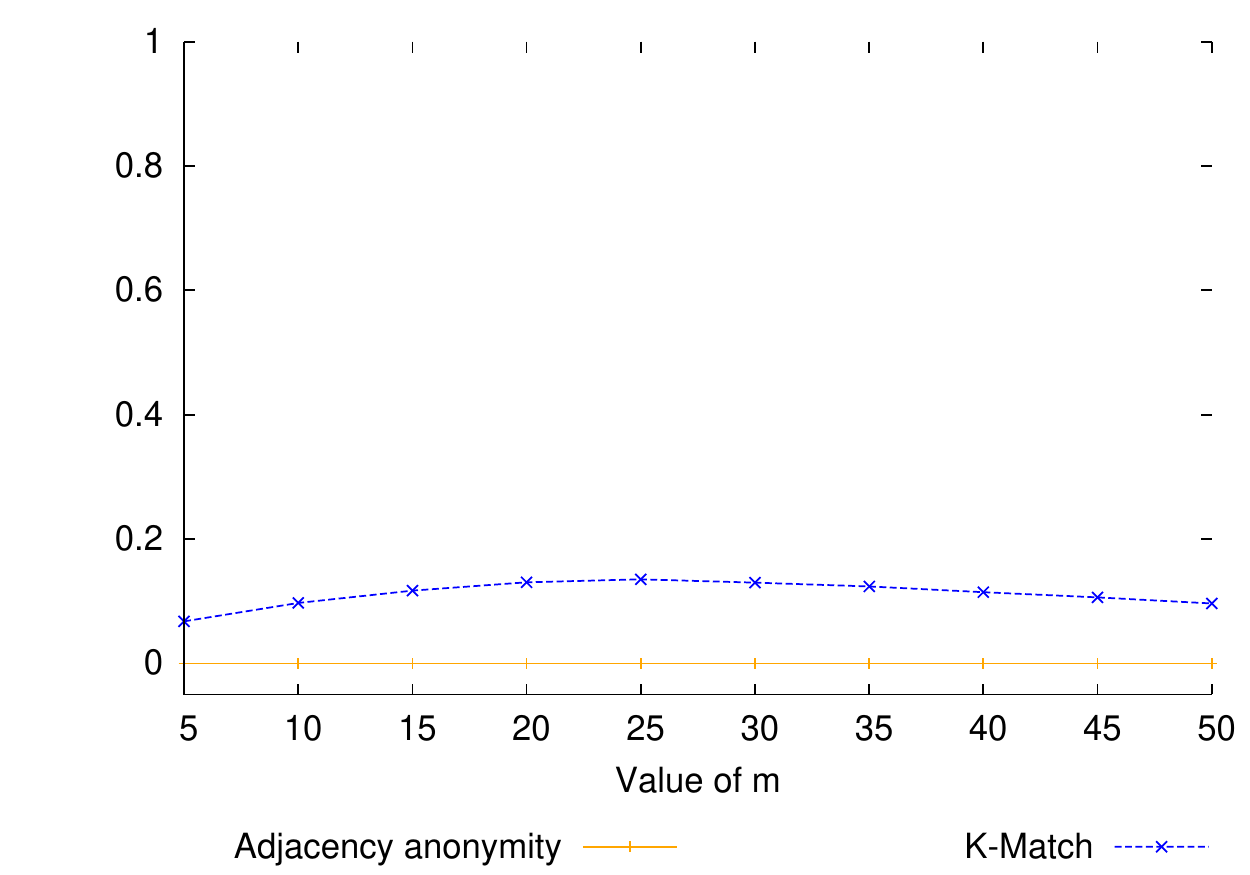}
}
\subfigure[t][ER graphs, $k=5$]{
\includegraphics[scale=.54]{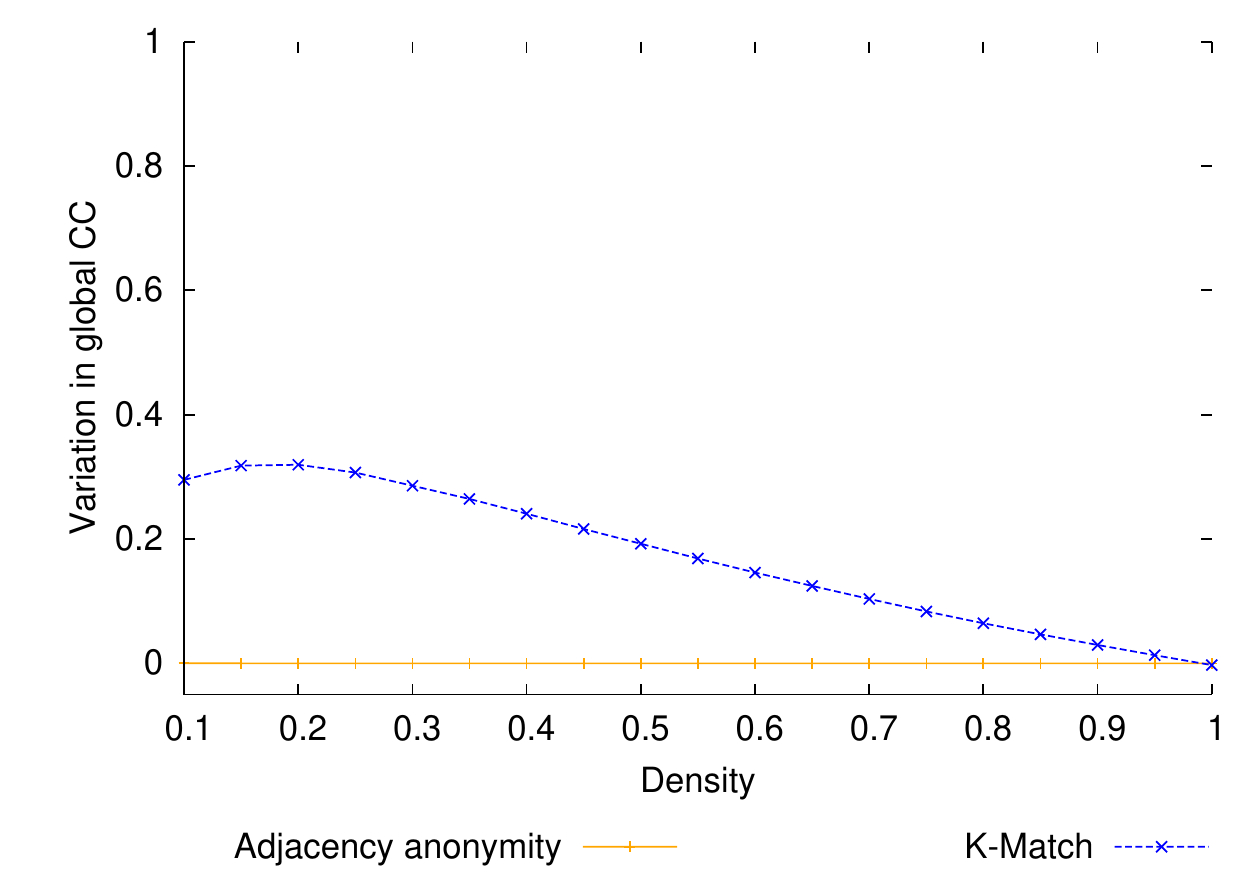}
}
\subfigure[t][BA graphs, $k=5$]{
\includegraphics[scale=.54]{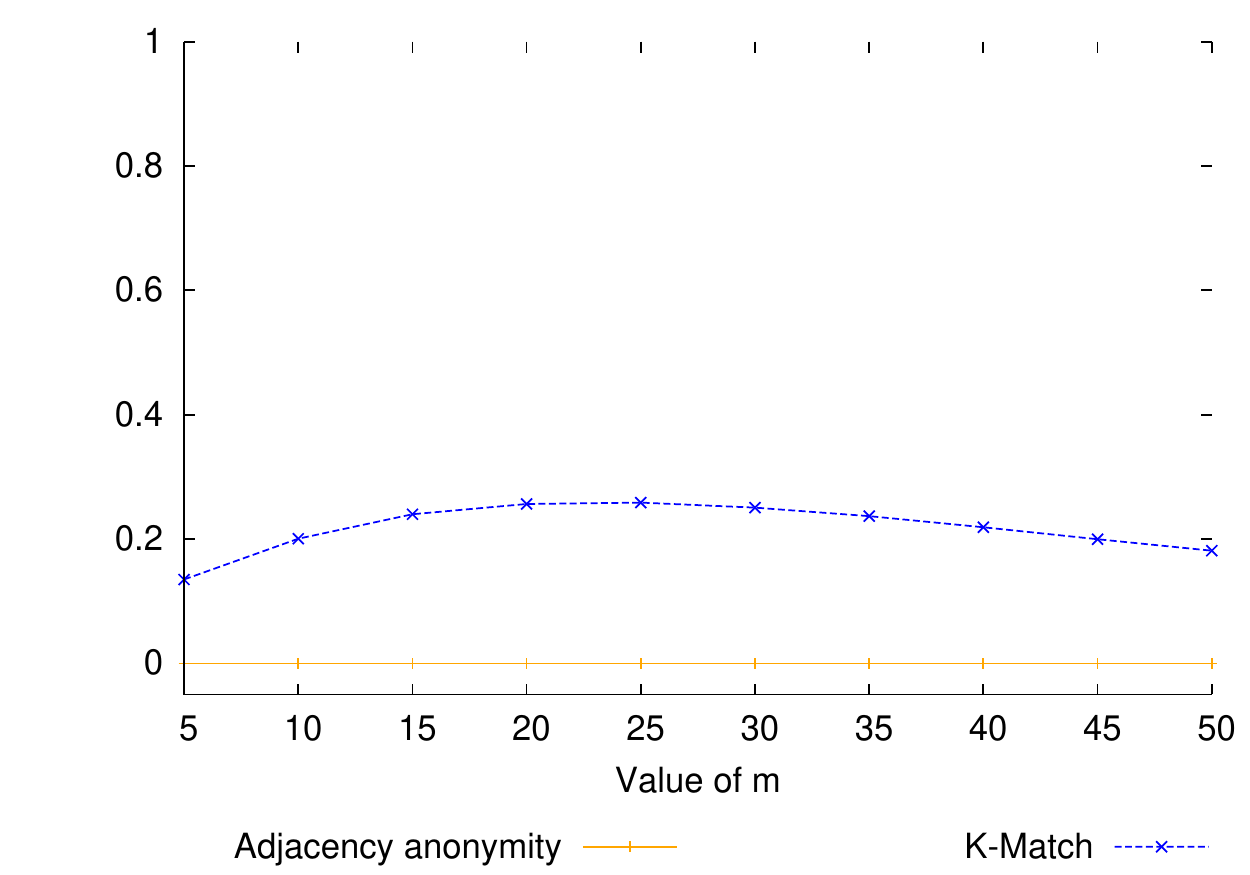}
}
\subfigure[t][ER graphs, $k=8$]{
\includegraphics[scale=.54]{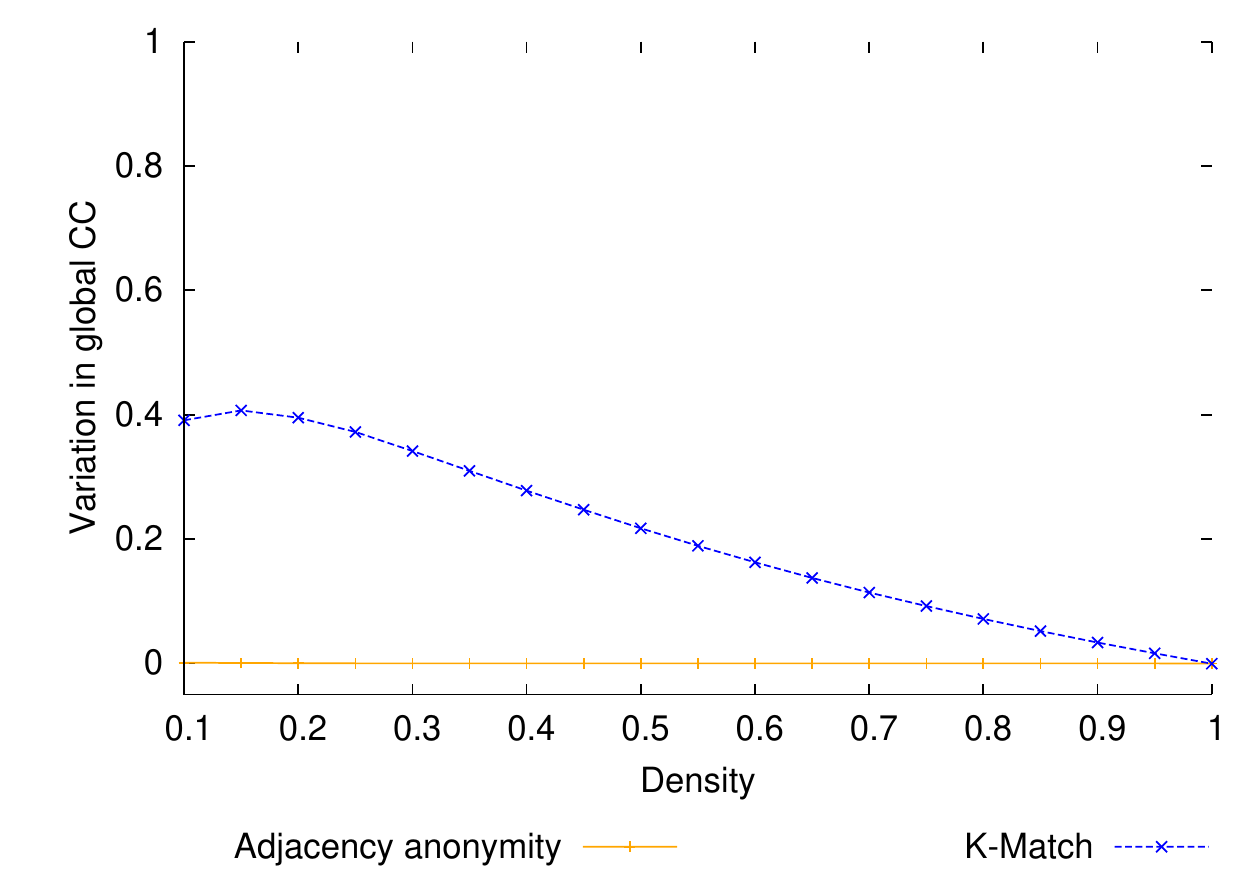}
}
\subfigure[t][BA graphs, $k=8$]{
\includegraphics[scale=.54]{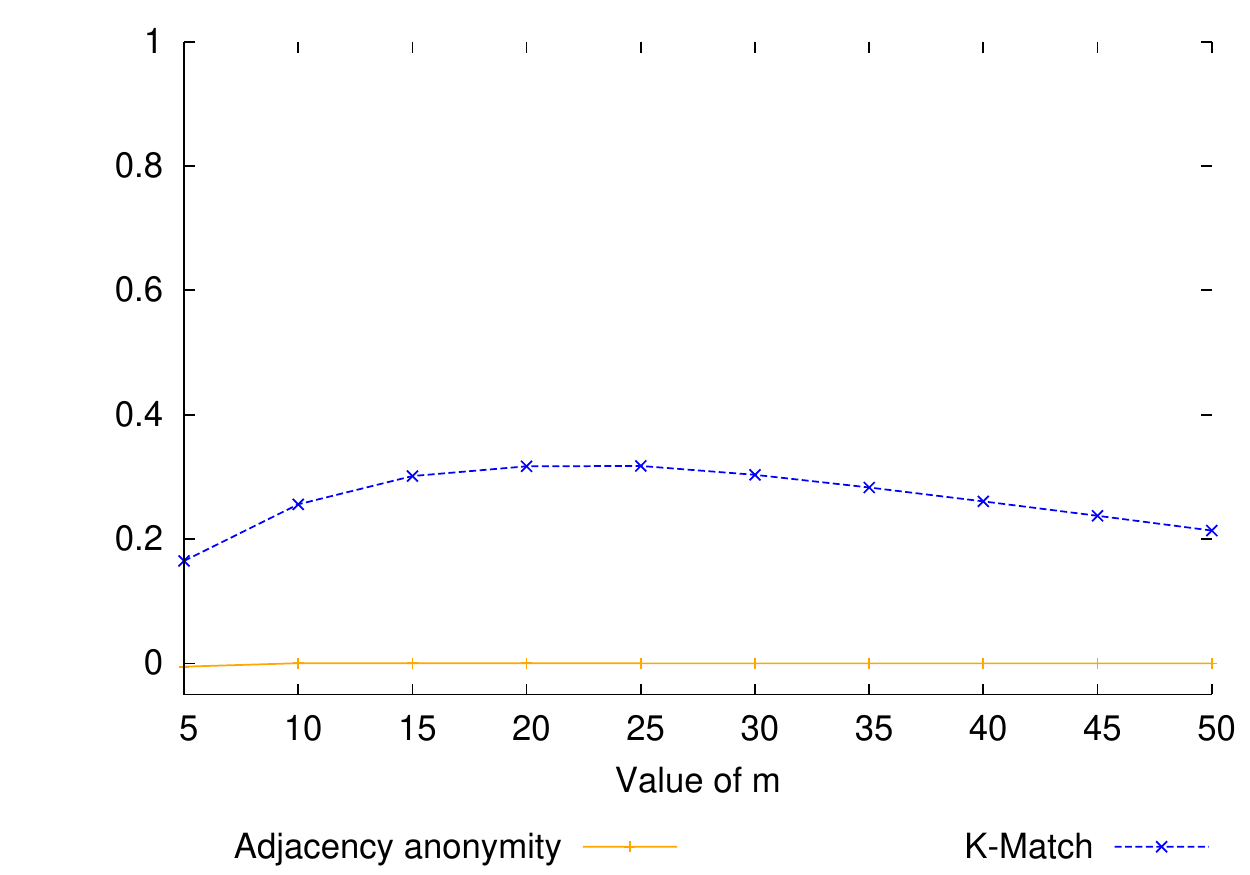}
}
\caption{Variations in global clustering coefficients on the collections 
of Erd\H{o}s-R\'enyi (left) and Barab\'asi-Albert (right) random graphs, 
with $\ell=8$ and $k\in\{2,5,8\}$.} 
\label{fig-gcc-er-ba-200-8-syb-at-1} 
\end{figure} 



\begin{figure}[H]
\centering
\subfigure[t][ER graphs, $k=2$]{
\includegraphics[scale=.54]{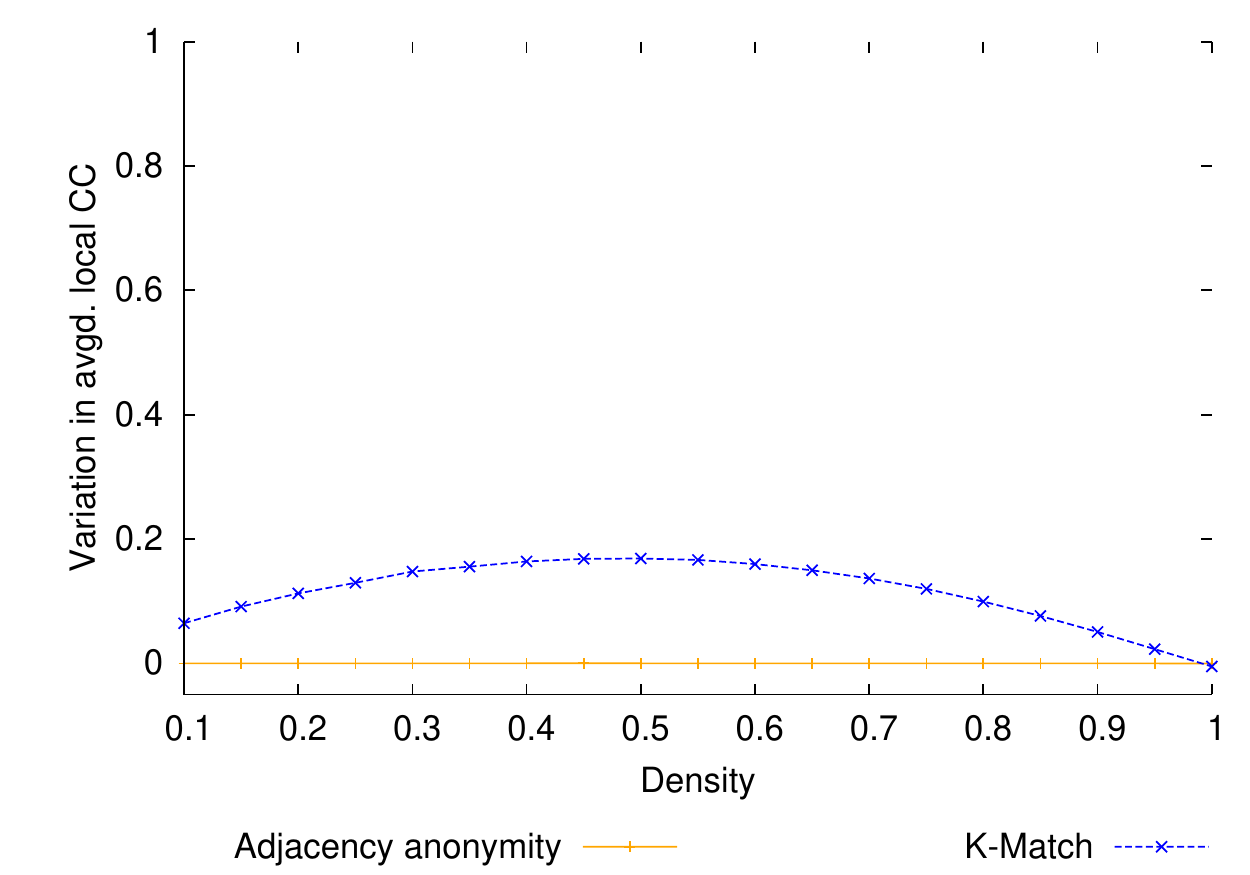}
}
\subfigure[t][BA graphs, $k=2$]{
\includegraphics[scale=.54]{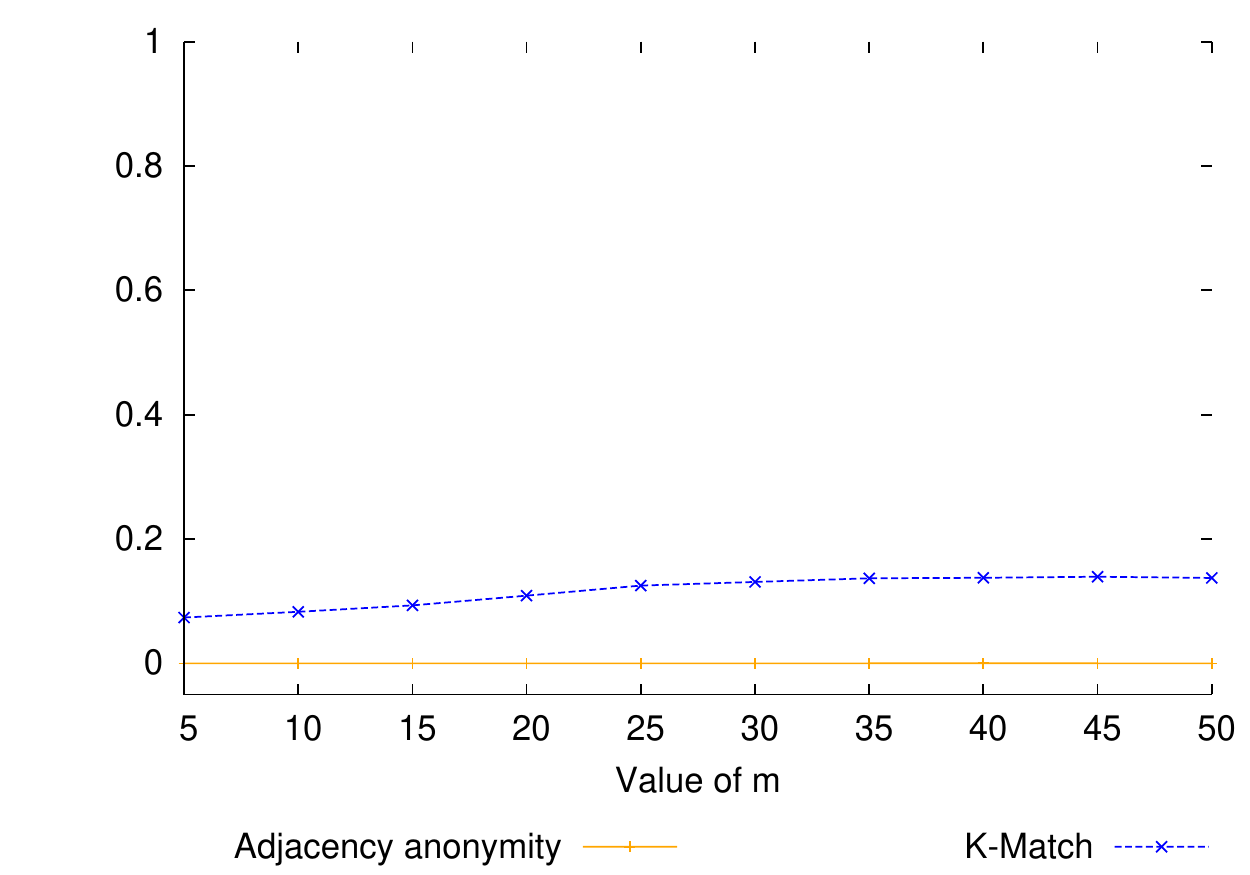}
}
\subfigure[t][ER graphs, $k=5$]{
\includegraphics[scale=.54]{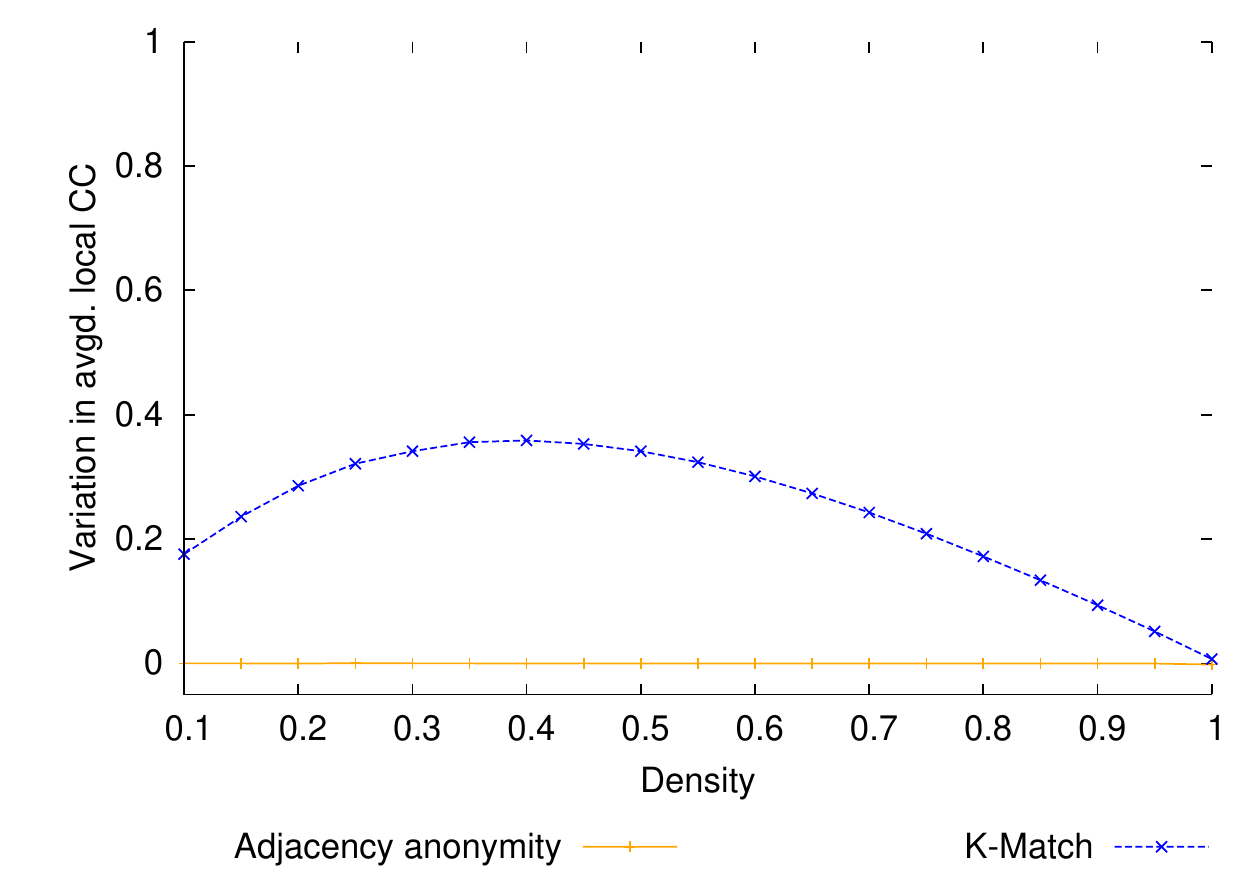}
}
\subfigure[t][BA graphs, $k=5$]{
\includegraphics[scale=.54]{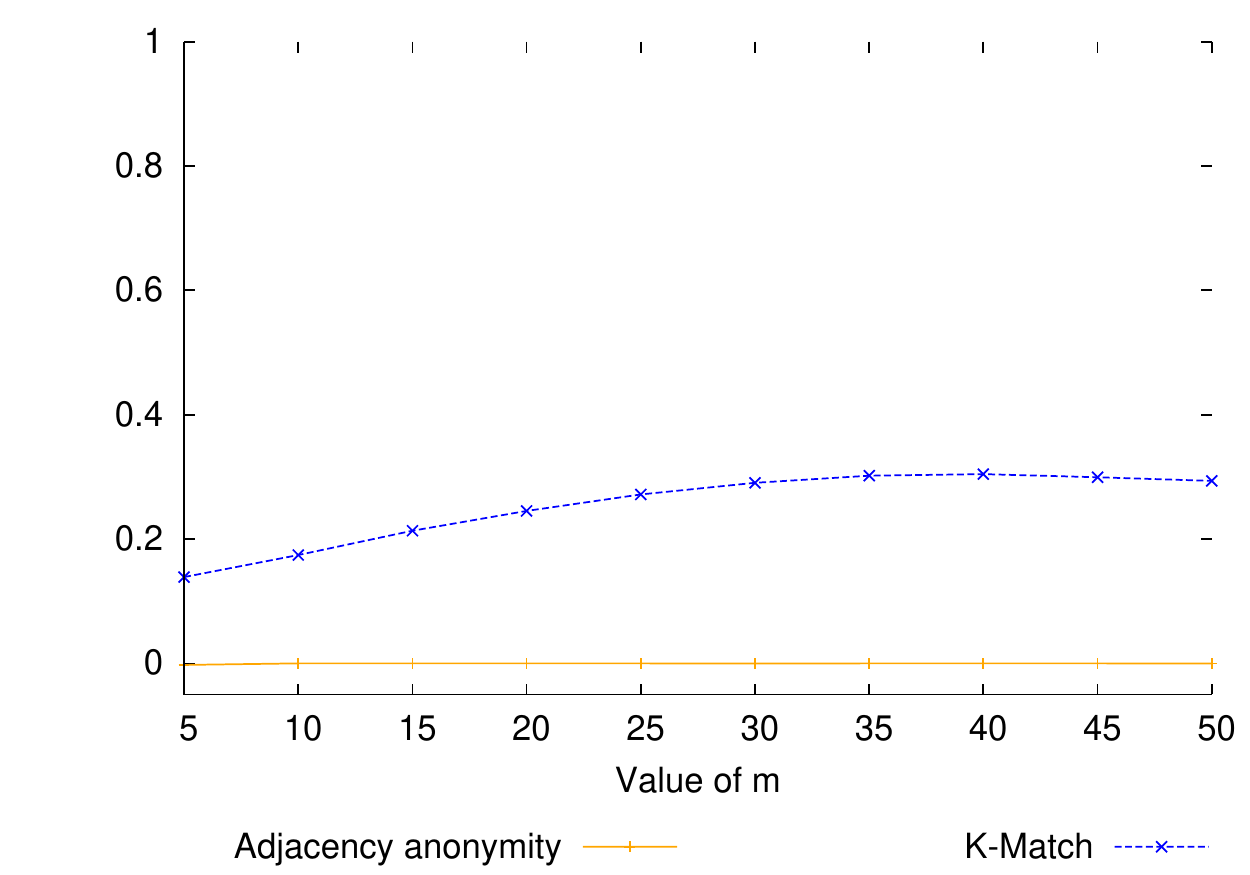}
}
\subfigure[t][ER graphs, $k=8$]{
\includegraphics[scale=.54]{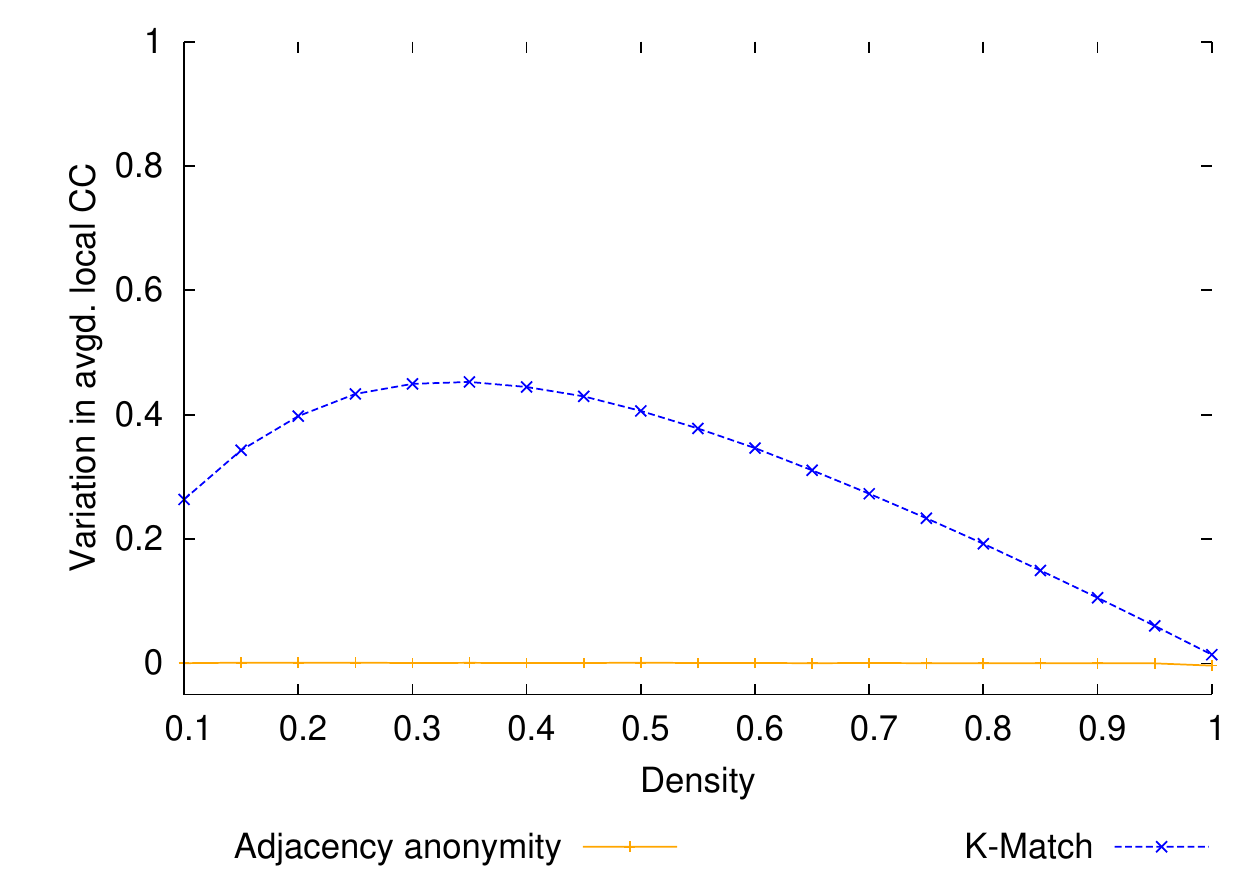}
}
\subfigure[t][BA graphs, $k=8$]{
\includegraphics[scale=.54]{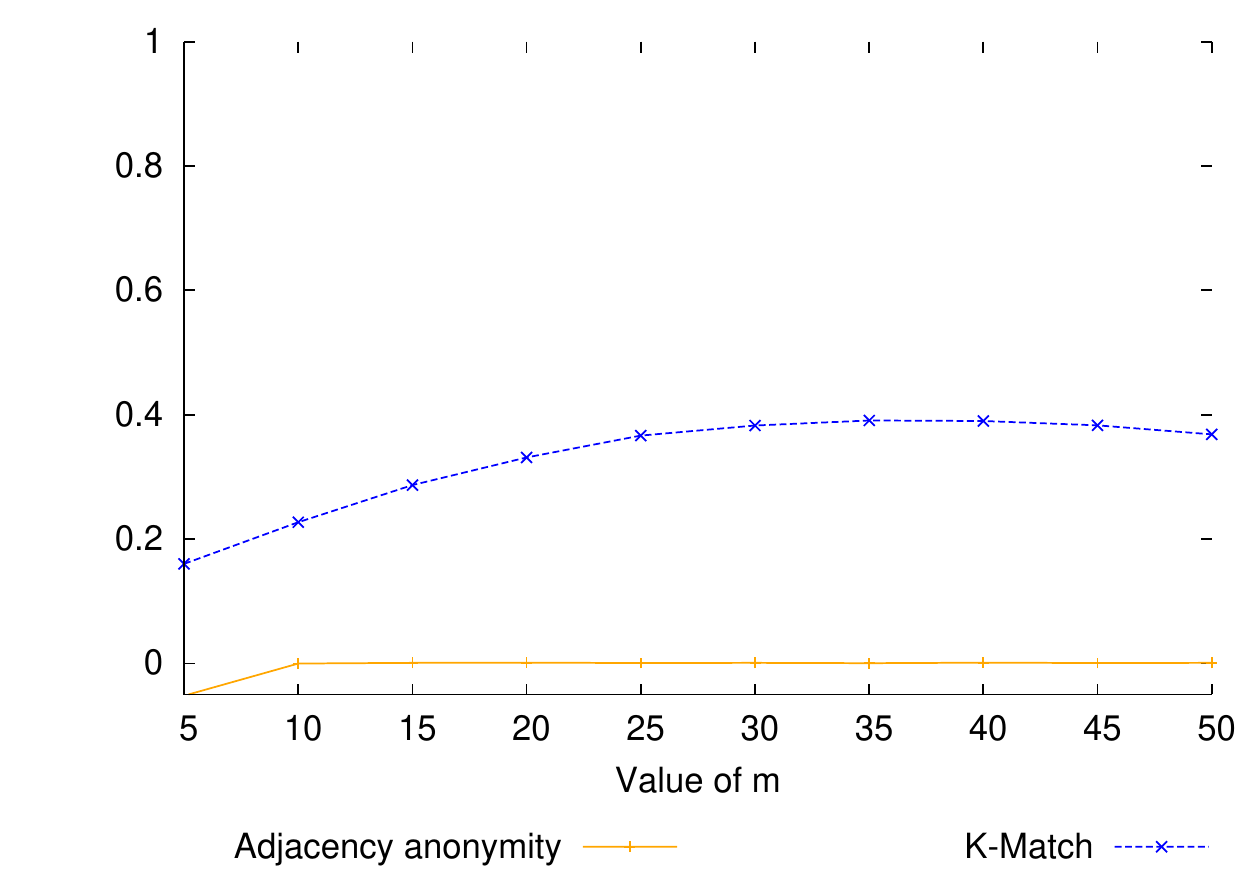}
}
\caption{Variations in averaged local clustering coefficients on the collections 
of Erd\H{o}s-R\'enyi (left) and Barab\'asi-Albert (right) random graphs, 
with $\ell=8$ and $k\in\{2,5,8\}$.} 
\label{fig-avg-lcc-er-ba-200-8-syb-at-1} 
\end{figure} 


\appendix
\section{Appendix}\label{app-automorphism}

It is claimed in~\cite{ZCO2009} that every vertex $v$ 
of a $k$-automorphic graph (see Definition~\ref{def-k-auto}) is structurally 
indistinguishable 
from $k-1$ other vertices $\varphi_1(v),\varphi_2(v),\ldots,\varphi_{k-1}(v)$. 

\begin{definition}[$k$-automorphism \cite{ZCO2009}]\label{def-k-auto}\ 
An \emph{automorphism} is an isomorphism from a graph to itself. 
Formally, an automorphism $\gamma$ within a graph $G=(V,E)$ is
a bijective function $\gamma \colon V \to V$, such that
$\forall v_1,v_2\in V \colon (v_1,v_2)\in E \iff 
(\gamma(v_1),\gamma(v_2))\in E$. 
A graph $G$ is said to be $k$-automorphic if there exist $k-1$ 
non-trivial automorphisms $\varphi_1, \varphi_2, \ldots, \varphi_{k-1}$ 
of $G$ such that $\varphi_i(v)\ne\varphi_j(v)$ for every $v\in V_G$ 
and every pair $i,j$ satisfying $1\le i<j\le k-1$.
\end{definition}

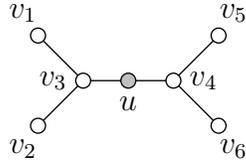
\begin{figure}[h]
	\centering
	\begin{tikzpicture}[inner sep=0.7mm, place/.style={circle,draw=black,
		fill=white},gray1/.style={circle,draw=black!99, 
		fill=black!75},gray2/.style={circle,draw=black!99, 
		fill=black!25},transition/.style={rectangle,draw=black!50,fill=black!20,thick},
		 line
		 width=.5pt]
	
	\coordinate (v) at (0,0);
	\coordinate (r) at (0.6,0);
	\coordinate (l) at (-0.6,0);
	
	\coordinate (ur) at (1.2,0.6);
	\coordinate (lr) at (1.2,-0.6);
	
	\coordinate (ul) at (-1.2,0.6);
	\coordinate (ll) at (-1.2,-0.6);
	
	\draw[black] (ul) -- (l) -- (v) -- (r) -- (ur); 
	\draw[black] (ll) -- (l);
	\draw[black] (lr) -- (r);
	
	\node[gray2] at (v) {};
	\node[place] at (r) {};
	\node[place] at (l) {};
	\node[place] at (ur) {};
	\node[place] at (lr) {};
	\node[place] at (ul) {};
	\node[place] at (ll) {};
	
	\coordinate [label=center:{$u$}] (lv) at (0,-0.3);
	\coordinate [label=center:{$v_1$}] (ull) at (-1.4,0.9);
	\coordinate [label=center:{$v_2$}] (lll) at (-1.4,-0.9);
	\coordinate [label=center:{$v_3$}] (ll) at (-1,0);
	\coordinate [label=center:{$v_4$}] (ll) at (1,0);
	\coordinate [label=center:{$v_5$}] (url) at (1.4,0.9);
	\coordinate [label=center:{$v_6$}] (lrl) at (1.4,-0.9);
	
	\coordinate [label=center:{\textcolor{white}{$v_0$}}] (lv0) at (-2.4,0);
	\coordinate [label=center:{\textcolor{white}{$v_{10}$}}] (lv10) at (2.5,0);
	
	
	\end{tikzpicture}
	\caption{A graph counterexample showing that $k$-automorphism 
	does not achieve the intended privacy protection.}\label{fig-automorphism}
\end{figure}

However, a missing condition in Definition~\ref{def-k-auto}, 
namely requiring every $\varphi_i$ to satisfy $\varphi_i(v)\neq v$, 
invalidates this claim. Consider the graph shown 
in Figure~\ref{fig-automorphism}. 
This graph satisfies $k$-automorphism as defined in Definition~\ref{def-k-auto}, 
as can be verified by the existence of the non-trivial automorphism 
$\gamma=\{(v_1,v_5), (v_2, v_6), (v_3, v_4), (u,u)\}$, yet 
the graph is vulnerable even to the simplest structural 
attack, the degree-based attack, as vertex $u$ is the 
sole vertex with degree $2$. 
It is worth noting that this limitation of $k$-automorphism 
does not necessarily invalidate existing anonymisation methods. 
This is exemplified by the \textsc{K-Match} algorithm itself, 
which does provide the intended protection because the property 
it directly enforces is the so-called \emph{$k$ different matches principle} 
(see~\cite{ZCO2009}), which in turn is not equivalent to $k$-automorphism, 
but stronger.

\end{document}